\documentclass{article}

\usepackage{microtype}
\usepackage{graphicx}
\usepackage{subcaption}
\usepackage{booktabs} 
\usepackage[most]{tcolorbox}
\usepackage{balance}
\usepackage{listings}
\usepackage{listingsutf8} 
\usepackage{tcolorbox}
\tcbuselibrary{skins,breakable,listings}

\lstdefinestyle{promptxml}{
  basicstyle=\ttfamily\footnotesize,
  columns=fullflexible,
  keepspaces=true,
  breaklines=true,
  breakatwhitespace=false,
  showstringspaces=false,
  inputencoding=utf8,
  escapeinside={(*@}{@*)},
  breakindent=2em,
  literate=
    {γ}{{$\gamma$}}1
    {δ}{{$\delta$}}1
    {α}{{$\alpha$}}1
}
\usepackage{hyperref}

\usepackage[accepted]{icml2026}

\usepackage{amsmath}
\usepackage{amssymb}
\usepackage{mathtools}
\usepackage{amsthm}
\usepackage{pifont}
\usepackage[table]{xcolor}
\usepackage{fontawesome5}
\usepackage{xcolor}
\usepackage{makecell}

\usepackage{thmtools} 

\usepackage{amsmath,amsfonts,bm}

\def\eqref#1{equation~\ref{#1}}

\def\1{\bm{1}}

\def\rvh{{\mathbf{h}}}

\def\rvn{{\mathbf{n}}}

\def\rvq{{\mathbf{q}}}

\def\rvx{{\mathbf{x}}}
\def\rvy{{\mathbf{y}}}

\DeclareMathAlphabet{\mathsfit}{\encodingdefault}{\sfdefault}{m}{sl}
\SetMathAlphabet{\mathsfit}{bold}{\encodingdefault}{\sfdefault}{bx}{n}

\newcommand{\R}{\mathbb{R}}

\newcommand{\sigmoid}{\sigma}

\DeclareMathOperator*{\argmax}{arg\,max}

\usepackage[capitalize,noabbrev]{cleveref}
\usepackage{nicefrac}
\usepackage{enumitem}
\usepackage{xurl} 
\usepackage{inconsolata} 
\usepackage[T1]{fontenc}

\theoremstyle{plain}
\newtheorem{theorem}{Theorem}[section]
\newtheorem{proposition}[theorem]{Proposition}
\newtheorem{lemma}[theorem]{Lemma}

\theoremstyle{definition}
\newtheorem{definition}[theorem]{Definition}

\theoremstyle{remark}
\newtheorem*{remark}{Remark} 

\usepackage[textsize=tiny]{todonotes}

\newcommand{\mypar}[1]{\noindent \textbf{#1}}

\newcommand{\ourshort}{PragLocker}
\newcommand{\ourcode}{\ourshort{}\textsubscript{code}} 
\newcommand{\ourtune}{\ourshort{}\textsubscript{tune}} 

\newcommand{\myalgcomment}[1]{\hfill \COMMENT{#1}}

\usepackage{multirow}
\usepackage{booktabs}
\usepackage{graphicx}

\icmltitlerunning{PragLocker: Protecting Agent Intellectual Property in Untrusted Deployments via Non-Portable Prompts}

\begin{document}

\twocolumn[
  \icmltitle{PragLocker: Protecting Agent Intellectual Property in Untrusted Deployments via Non-Portable Prompts}



\icmlsetsymbol{equal}{*}
\icmlsetsymbol{mail}{\faEnvelope[regular]}

\begin{icmlauthorlist}
  \icmlauthor{Qinfeng Li}{zju,equal}
  \icmlauthor{Yuntai Bao}{zju,equal}
  \icmlauthor{Jianghui Hu}{cau,equal}
  \icmlauthor{Wenqi Zhang}{ngic}
  \icmlauthor{Jintao Chen}{zju}
  \icmlauthor{Huifeng Zhu}{wustl}
  \icmlauthor{Yier Jin}{ustc}
  \icmlauthor{Xuhong Zhang}{ngic,mail}
\end{icmlauthorlist}

\icmlaffiliation{zju}{Zhejiang University}
\icmlaffiliation{ngic}{Innovation and Management Center, School of Software Technology (Ningbo), Zhejiang University}
\icmlaffiliation{ustc}{University of Science and Technology of China}
\icmlaffiliation{cau}{Chang'an University}
\icmlaffiliation{wustl}{Washington University in St. Louis}

\icmlcorrespondingauthor{Xuhong Zhang}{zhangxuhong@zju.edu.cn}
  \icmlkeywords{Machine Learning, ICML}

  \vskip 0.3in
]



\printAffiliationsAndNotice{\icmlEqualContribution. \quad \faEnvelope[regular] Corresponding author.}  

\begin{abstract}
LLM agents rely on prompts to implement task-specific capabilities based on foundation LLMs, making agent prompts valuable intellectual property. However, in untrusted deployments, adversaries can copy and reuse these prompts with other proprietary LLMs, causing economic losses. To protect these prompts, we identify four key challenges: proactivity, runtime protection, usability, and non-portability that existing approaches fail to address. We present PragLocker, a prompt protection scheme that satisfies these requirements. PragLocker constructs function-preserving obfuscated prompts by anchoring semantics with code symbols and then using target-model feedback to inject noise, yielding prompts that only work on the target LLM. Experiments across multiple agent systems, datasets, and foundation LLMs show that PragLocker substantially reduces cross-LLM portability, maintains target performance, and remains robust against adaptive attackers.

\begin{center}
\faGithub ~\href{https://github.com/ZJU-OmniAI/praglocker}{Code}
\end{center}
\end{abstract}

\section{Introduction}
In recent years, amid the rapid adoption of large language models (LLMs), intelligent LLM agents, such as Cursor~\cite{cursor_features}, Manus ~\cite{ManusAI_Website_2026}, and Zapier~\cite{zapier_agents}, have emerged as autonomous executors of complex tasks. These agents, serving as a key interface between LLM capabilities and real-world applications, implement task capabilities largely through an agent \emph{system prompt} that specifies tasks, policies, and tool-use on top of a foundation LLM~\cite{openai_text_guide_instructions,anthropic_system_prompts}. As a result, prompt design becomes a key determinant of agent behavior and performance, constituting a core competitive asset. In particular, many agents may invoke the same underlying LLM (e.g., proprietary ones like ChatGPT or Gemini)~\cite{openai_chatgpt_2022,gemini_2023}, yet yield substantially different functionality and quality.

However, as highly valuable intellectual property (IP), prompts are vulnerable to theft after deployment, leading to substantial economic losses. Specifically, crafting a high-quality agent prompt typically requires significant expert knowledge and continual real-world iteration, making it a highly valuable asset~\cite{sahoo2024prompt_survey,yang2025prsa}. However, in practice, agents often run on user devices~\cite{cursor_features}, cloud services~\cite{spector2025zapier_agents_guide}, or multi-tenant infrastructures~\cite{nist_sp800_145}, where adversaries (e.g., malicious end users or insider cloud operators) may copy and misuse the prompt~\cite{hui2024pleak,wang2024raccoon}. Worse still, prompts are typically written in natural language; once leaked, they can be easily reused on any other, even stronger, proprietary LLM.
\textit{Once leaked, an attacker can reuse the prompt with other LLMs to build a similar or even stronger agent, undermining the original agent’s competitive advantage and causing substantial losses.} Therefore, effectively protecting these prompts in the deployment of agent systems has become a critical issue.

\begin{table*}[ht]
\centering
\caption{Comparison with existing solutions. \checkmark/\ding{55} illustrates whether the method can achieve the corresponding property. 
}
\resizebox{2\columnwidth}{!}{
\begin{tabular}{lcccc}
\toprule
Solutions (exemplar)                  &  Proactivity     &  Runtime security  & Usability & Non-Portability \\ \midrule
Prompt Watermarking~\cite{yang2025promptcos}              & \ding{55}          & \checkmark                    & \checkmark & \ding{55} \\
Encryption-based Protection~\cite{kubernetes_encrypt_data_at_rest})        & \checkmark      & \ding{55} & \checkmark  & \ding{55}  \\
Prompt Obfuscation~\cite{pape2025prompt}              & \checkmark      & \checkmark    & \ding{55}  & \ding{55} \\
PragLocker~(ours)              & \checkmark      & \checkmark   & \checkmark     & \checkmark       \\ \bottomrule
\end{tabular}
}
\label{table_intro}
\end{table*}

Unfortunately, as shown in Table \ref{table_intro}, traditional solutions struggle to protect the prompt of untrusted environment-deployed agents as they fail to satisfy the diverse requirements.
First, passive protection methods, e.g., \textbf{prompt watermarking}~\cite{yang2025promptcos,yao2024promptcare}, primarily verify ownership \textit{after} misuse. Thus, they do not proactively prevent prompt theft, rendering these methods vulnerable: once a prompt is stolen, it may be freely exploited without detection. 
In contrast, proactive protections aim to prevent unauthorized usage. For example, \textbf{encryption-based protection}~\cite{rana2023kms_survey,kubernetes_encrypt_data_at_rest}) keeps confidential information encrypted during transmission, distribution, and storage. However, at runtime, the prompt must be submitted to black-box LLM API~\cite{openai_api_reference_introduction,google_gemini_api_models} in plaintext; therefore, encryption can only protect the prompt at rest, not during inference.

Alternatively, a potential method is \textbf{prompt obfuscation}~\cite{pape2025prompt}, which replaces the original system prompt with an obfuscated yet function-preserving variant.
However, existing prompt obfuscation techniques are not directly applicable to agent prompt protection: they either remain portable across LLMs or require white-box access to the underlying model.
Specifically, EmojiPrompt~\cite{lin2025emojiprompt} encodes prompts into emojis to degrade human readability, but other LLMs can still decode and follow such surface-level encodings, so it does not prevent cross-model reuse. A closer line of work constructs prompts that appear noise-like and less portable while preserving utility on a target model; e.g., \citet{pape2025prompt} obfuscates prompts via reversible semantic transformations in the model’s representation space. However, many real-world agents (e.g., Cursor and Copilot) rely on proprietary black-box LLMs, where developers only have API-level access, making such white-box methods inapplicable. Fundamentally, \emph{constructing a prompt that is usable on the target LLM yet non-portable is challenging}, especially in black-box settings, where one can only rely on input-output feedback. The problem is further exacerbated by intra-family transfer: intra-family LLMs often behave similarly~\cite{mcgovern2025fingerprints}, further weakening the behavioral separation needed for non-portability.

Considering the limitations of existing defense strategies, we identify four key challenges (C) in protecting prompts for agents deployed in untrusted environments. \textbf{C1} (Proactivity): ensuring the prompt cannot be misused even if physically obtained by an attacker. \textbf{C2} (Runtime protection): ensuring protection not only during deployment but also during inference. \textbf{C3} (Usability): the protected prompt must preserve the agent’s performance on the target LLM.
\textbf{C4} (Non-Portability): simultaneously, the protected prompt must be human-unintelligible and ineffective on other LLMs.

In this paper, we propose \textbf{PragLocker} (\textbf{Ag}ent \textbf{Pr}ompt \textbf{Locker}), a black-box prompt obfuscation scheme. PragLocker obfuscates the system prompt before deployment, so the original prompt is never released in plaintext (addressing \textbf{C1}). This obfuscated prompt remains behavior-preserving on the target LLM and can be used directly without deobfuscation (addressing \textbf{C2}).

PragLocker is grounded in a key question: \textit{Does there exist such an obfuscated prompt that preserve same behavior on a target LLM while failing on other LLMs?} We answer yes with a theoretical motivation: there exists a perturbed prompt that makes the target LLM produce the same next token while causing a different next token on other LLMs.

Despite our theoretical insights, practical construction remains challenging. To obtain such prompts under the black-box constraint, PragLocker proceeds in two phases. First, it performs an \textit{initialization transformation} that converts the prompt into a code-symbol form, yielding an initial non-natural-language, utility-preserving concealment. However, this transformed prompt does not yet satisfy our goals: it can still be understood by other LLMs (many models can interpret code and symbols).
To address this issue, PragLocker introduces a second phase, \textit{noise-injected prompt optimization}. Specifically, analogous to black-box discrete optimization, we progressively inject character-level noise into the prompt as the optimization to reproduce the target LLM's original outputs. 
As a result, the noise is guided solely by target-LLM feedback; the resulting obfuscated prompt becomes model-specific: it preserves performance on the target LLM (addressing \textbf{C3}) while appearing noise-like and ineffective on other LLMs (addressing \textbf{C4}).

We evaluate PragLocker across diverse agent systems, foundation LLMs, and tasks, demonstrating strong protection without sacrificing task performance. Our obfuscated prompts exhibit near-zero portability across LLMs, even between FP16 and 4-bit variants of the same LLM. Moreover, \textit{even target LLM itself} cannot reliably interpret the obfuscated prompt, which instead behaves like a \emph{model-conditioned trigger} without retaining recoverable text-level information. The contributions of this work are as follows:
\begin{itemize}[topsep=-4pt,itemsep=-1.7pt]
\item To our knowledge, we are the first to study prompt protection for agents under untrusted deployment, where adversaries may copy and obtain the deployed prompt. We identify prompts as the primary embodiment of agent IP and distill four requirements for protection. 
\item We propose PragLocker, a black-box prompt protection scheme that replaces the system prompt with an obfuscated prompt that works on the target LLM but fails on others. We further provide a theoretical motivation by proving an existence theorem for such model-specific, function-preserving obfuscated prompts.
\item To construct such prompts under pure black-box access, PragLocker uses a two-phase pipeline: an \textit{initialization transformation} followed by \textit{noise injection}, optimizing obfuscated prompts without model internals.
\item Experiments across agent systems, datasets, and LLMs show that PragLocker preserves target performance, sharply reduces cross-model portability, and withstands adaptive attacks. We also report ablations and case studies to validate key design choices.
\end{itemize}

\section{Preliminaries}
\subsection{LLM Agent}
LLM agents~\cite{wang2024survey_llm_agents} are task-oriented systems built on top of foundation LLMs. Concretely, an agent is mainly defined by (i) a \emph{system prompt} that specifies tasks, policies, and tool-use, and (ii) an underlying LLM that executes prompt-conditioned reasoning and generation. In practice, many commercial agents~\cite{cursor_features,friedman2021copilot} are built on the same proprietary black-box foundation models, \emph{making the system prompt the dominant portion of the agent developer's IP}. In addition, agents are frequently deployed in heterogeneous and potentially untrusted environments~\cite{nist_sp800_145,spector2025zapier_agents_guide}, including third-party hosting platforms, multi-tenant clouds, and end-user devices, which increases the exposure surface of the deployed prompt and elevates prompt confidentiality and integrity as central security concerns.

\subsection{Related Work}

\mypar{Prompt Watermarking.} Prompt watermarking embeds a signature into a prompt to provide ownership evidence, enabling verification after theft. For example, PromptCARE~\cite{yao2024promptcare} applies small, utility-preserving perturbations and detects infringement via a predefined verification protocol (e.g., trigger queries) that yields statistically distinguishable response patterns. More recent work (e.g., PromptCOS~\cite{yang2025promptcos}) moves toward content-only auditing by jointly optimizing the system prompt with verification queries and target “signal” outputs, then testing suspected prompts via output-similarity based verification.

\mypar{Encryption-based Protection.} Conventional encryption (e.g., storing the system prompt as an encrypted file/secret~\cite{kubernetes_encrypt_data_at_rest}) protects confidentiality at rest and in transit, but inference requires decryption into plaintext and transmission to the LLM, allowing the execution host to extract the prompt post-decryption~\cite{karvandi2024reversing}. TEEs offer attested isolation (e.g., SGX/TDX~\cite{costan2016intel}, GPU confidential computing~\cite{nvidia_hcc_whitepaper_wp11459}) to reduce runtime exposure, yet are often incompatible with LLM agents: many rely on proprietary black-box API LLMs~\cite{openai_api_reference_introduction, google_gemini_api_models}, so the model cannot run inside the same trusted environment and the decrypted prompt must still be sent to external services, undermining confidentiality.

\mypar{Confidentiality Infrastructure.}
Confidential-inference infrastructure protects prompt privacy by providing an attested TEE/CVM endpoint so prompts are decrypted and processed only inside a protected runtime, sometimes with split/partitioned execution to maintain serving efficiency~\cite{gim2024confidential,yuan2025scx,su2025runtime}. However, it requires third-party platform cooperation (attestation/key release, correct routing) and thus is outside our threat model: a malicious platform can still access the prompt in plaintext.

\mypar{Prompt Obfuscation.} Prompt obfuscation replaces the original system prompt with an obfuscated yet function-preserving variant. For example, EmojiPrompt~\cite{lin2025emojiprompt} encodes sensitive text into emoji-based representations. However, such surface-level obfuscation does not prevent \emph{reuse on other LLMs}, offering no defense against cross-model misuse. \citet{pape2025prompt} optimizes prompts in the model’s representation space to retain utility while concealing the underlying instructions. These methods, however, require model-internal access, which is unavailable in black-box agent deployments.

\mypar{Prompt Optimization.}
\textit{Soft prompt tuning} is a well-founded approach to parameter-efficient fine-tuning, which trains task-specific embeddings as prompt prefix to task queries using gradient-based optimization~\citep{lester2021power}.
Historically, discrete prompt optimization is a difficult problem~\citep{shin2020autoprompt,singh2023explaining} and gradient-free optimization is even harder~\citep{deng2022rlprompt,zhang2024agent}.
We position our technical objective as gradient-free discrete optimization that optimizes for task performance, obfuscation, and non-portability constraints.

\subsection{Threat Model}
We consider two parties: the defender is the party that owns the deployed agent, and the attacker aims to steal the prompt.

\mypar{Defender.} The defender deploys an LLM agent as a service (e.g., on third-party hosting or end-user devices) and accesses the target LLM only via a black-box API. The target LLM is trusted not to exfiltrate prompts, while the deployment environment is untrusted. The defender’s goal is to proactively provide runtime protection for the agent system prompt while preserving on-target usability, such that even if the deployed prompt is obtained, it cannot be reused to reproduce comparable behavior on non-target LLMs.

\mypar{Adversary.} The adversary can obtain the deployed prompt. An attack succeeds if the stolen (or recovered) prompt enables similar agent functionality on a different model. Beyond naïve copying, the adversary may use limited compute and data to mount adaptive recovery attacks (e.g., deobfuscation optimization) to reconstruct a usable prompt.

\section{Our Design: PragLocker}

\begin{figure*}[t]
    \centering
    \includegraphics[width=1\linewidth]{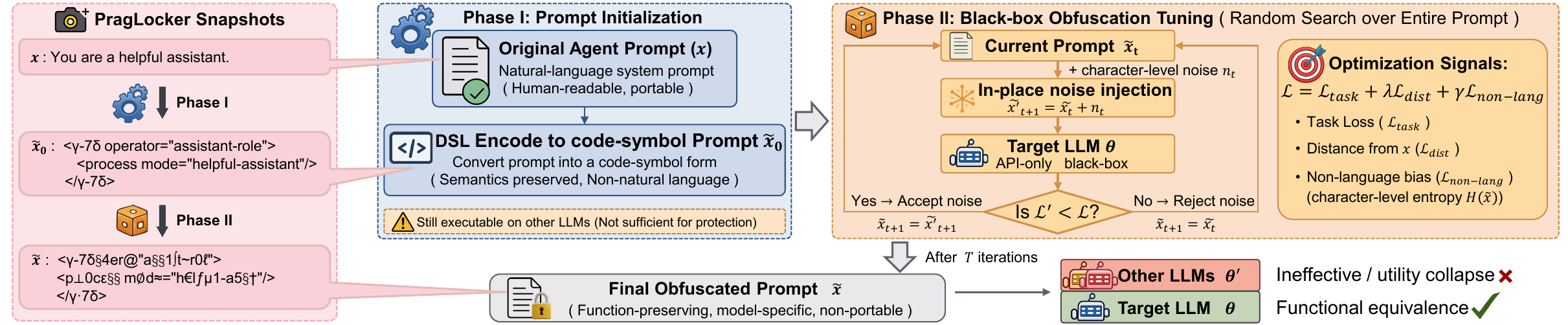}
    \caption{A pipeline of PragLocker. PragLocker transforms a plaintext prompt into a model-specific obfuscated form through a two-phase process: (i) a code-symbol initialization that preserves task semantics, and (ii) noise-injected black-box optimization driven by target LLM feedback. The final obfuscated prompt remains usable at runtime on the target LLM but resists reuse or recovery on other models.}
    \label{method}
\vskip -0.1in
\end{figure*}



\subsection{Problem Formulation}
Let the prompt space be $\mathcal{P} \subset \{ \mathcal{V}, \mathcal{V}^2, \dots \}$ where $\mathcal{V}$ is model vocabulary,
the protected agent prompt be $\rvx \in \mathcal{P}$
and $\theta$ be model parameters.
\textbf{Our objective} is to identify another prompt $\tilde{\rvx} \in \mathcal{P}$ that satisfy the following criteria:
(1) \textit{Obfuscation (C2)}: $\tilde{\rvx}$ is distant from $\rvx$ in prompt space;
(2) \textit{Usability (C3)}: $\tilde{\rvx}$ is functionally equivalent to $\rvx$;
(3) \textit{Non-portability (C4)}: the utility objective is satisfied only on the target LLM, not other LLMs.

\subsection{Theoretical Motivation} \label{subsec:theoretical_motivation}
In this subsection, we explain the theoretical motivation for the feasibility of \ourshort{}, i.e. to identify an alternative prompt to the original prompt that simultaneously satisfy the criteria of obfuscation, usability and non-portability. 
We defer detailed discussions to Appendix \cref{app:theoretical_motivation}.

\mypar{Roadmap.}
We define \textit{functional equivalence} for single-token generation, prove the existence of \textit{local stability regions} that allow for obfuscation, and discuss \textit{non-portability} across models.

\mypar{Local Obfuscation and Utility.} Let $\mathbf{h} = \mathrm{Embed}(\mathbf{x}) \in \mathbb{R}^{n \times d}$ be the embeddings for prompt $\rvx$. We seek an obfuscated prompt $\tilde{\rvx}$ with embeddings $\tilde{\rvh} = \rvh + \bm{\delta}$ such that output behavior is preserved under greedy decoding.

\begin{definition}[Functional equivalence]
Embeddings $\tilde{\rvh}$ and $\rvh$ are equivalent w.r.t. query $\rvq_i$ if $\argmax_y f(y \vert \tilde{\rvh}, \rvq_i) = \argmax_y f(y \vert \rvh, \rvq_i)$.
\end{definition}

Trivially, equivalence holds if the \textit{correct-class margin} $m(\tilde{\rvh}, \rvq_i, y_i) \coloneqq f(\tilde{\rvh}, \rvq_i)_{y_i} - \max_{k \neq y_i} f(\tilde{\rvh}, \rvq_i)_k$ remains positive. Since $f(\cdot;\theta)$ is a composition of continuous functions (layers, activations), $m(\cdot)$ is continuous. Thus, for any $\rvh$ where $m > 0$, there exists an $\epsilon$-ball $B_{\epsilon}(\rvh)$ within a \textit{stability region} $S_{\rvx} = \{ \rvh' \mid \forall \rvq_i \in \mathcal{Q}, m(\rvh', \rvq_i, y_i) > 0 \}$ where utility is perfectly preserved.

\begin{theorem}[Existence of obfuscated prompts]
For a prompt $\rvx$, there exists $\tilde{\rvx} \neq \rvx$ such that $\tilde{\rvh} \in S_{\rvx}$ (Utility) and $d(\tilde{\rvx}, \rvx) \geq d_0$ (Obfuscation).
\end{theorem}
\textit{Proof sketch.} Transformer attention often exhibits low sensitivity to specific tokens (attention dilution). By perturbing $k$ such tokens, the cumulative shift $\|\Delta \rvh\| \leq \sum_{j \in \mathcal{K}} \|\bm{\delta}_j\|$ can be kept within $\epsilon$ to maintain utility, while the distance $d(\tilde{\rvx}, \rvx)$ grows with $k$ to satisfy the obfuscation bound $d_0$.

\mypar{Non-portability.} Non-portability arises from \textit{manifold mismatch}: the stability region $S_{\rvx}(\theta)$ is highly dependent on the specific geometry of the loss landscape defined by $\theta$. A perturbation $\bm{\delta}$ optimized to stay within the decision boundaries of model $\theta$ is unlikely to reside within the distinct stability region $S_{\rvx}(\theta')$ of a different model, as $P(\tilde{\rvh} \in S_{\rvx}(\theta')) \ll P(\tilde{\rvh} \in S_{\rvx}(\theta))$ in high-dimensional space.

\begin{remark}[Between our theoretical motivation and method design]
Although our methodology is not entirely aligned with our discussions above,
we point out that these analyses serve as a theoretical lens that motivates our method design, showcasing the possibility of identifying prompts that satisfy the obfuscation, utility, and non-portability criteria.
\end{remark}

\subsection{\ourshort{} Methodology}

\mypar{Overview.}
As is shown in \cref{method}, \ourshort{} consists of two components:
prompt initialization and obfuscation optimization.
Prompt initialization primarily contributes to our utility objective, while obfuscation optimization is used to achieve obfuscation and non-portability objectives.

\mypar{Prompt Initialization.}
Our prompt initialization strategy consists of two key design considerations.
First, we use the original prompt $\rvx$ as the warm-start, which is essential for ensuring that the obfuscated prompt achieves the same level of performance as the original task prompt.
Second, we prompt the target LLM to to encode $\rvx$ in a code-symbol form, such that the encoded prompt ($\tilde{\rvx}_0$) is no longer written in natural language while having the same semantics as $\rvx$.
This approach could be understood as a preliminary obfuscation step conditioned by the target LLM that also creates redundancy, thereby setting the stage for subsequent obfuscation optimization.



\mypar{Non-Portable Obfuscation Optimization.}
Our obfuscation strategy is built upon \textit{random search (RS)}, a black-box, gradient-free, discrete optimization method~\citep{rastrigin1963convergence}.
This approach is reminiscent of \textit{evolution strategies}~\citep{schwefel1977evolutionsstrategien,salimans2017evolution}, and has recently been used for gradient-free optimization of jailbreaking suffixes~\citep{andriushchenko2025jailbreaking}.

RS allows us to optimize textual prompts when there is only API-level access to prompt inputs, responses, and log-probabilities of the target LLM.
The \textbf{intuition} for RS to work well in practice is that, although it is convenient for humans to describe task-solving requirements and instructions in natural language, it might not be the most effective and efficient approach~\citep{chang2024efficient,li2025prompt}.
Therefore, the inherent redundancy of the initialized prompt ($\tilde{\rvx}_0$) allows RS to improve model performance while leaving room for obfuscation.

We provide an algorithmic description of obfuscation optimization in \cref{alg:praglocker}.
Obfuscation optimization might be understood as RS over the entire task prompt $\rvx$.
At each step of obfuscation optimization, we inject textual noise from a predesignated noise set, usually the commonly used set of printable characters.
Since each instance of noise might have a mixed effect on either task performance or textual readability, we constrain the accepted set of noise with a custom objective function to ensure that the final obfuscated prompt is a piece of non-natural text that retains task performance.

\mypar{Optimization Objective.}
The objective function for obfuscation optimization consists of three terms:
\begin{equation}
\begin{aligned}
\mathcal{L} &= \mathcal{L}_{\text{task}} + \lambda \mathcal{L}_{\text{dist}} + \gamma \mathcal{L}_{\text{non-lang}}, \text{~where} \\
\mathcal{L}_{\text{task}} &= - \log p(\rvy \vert \rvq, \tilde{\rvx}), \\
\mathcal{L}_{\text{dist}} &= - \log \sigmoid (\mathrm{Dist}(\tilde{\rvx}, \rvx)), \\
\mathcal{L}_{\text{non-lang}} &= - H(\tilde{\rvx}),
\end{aligned}
\end{equation}
where $\lambda, \gamma \in \R$ are constant coefficients and $\sigma(\cdot)$ is the sigmoid function;
$\mathcal{L}_{\text{task}}$ implements the utility constraint,
$\mathcal{L}_{\text{dist}}$ implements the obfuscation constraint
while $\mathcal{L}_{\text{non-lang}}$ implements both obfuscation and non-portability constraints.

Particularly, $\mathcal{L}_{\text{dist}}$ ensures that the final obfuscated prompt differs from the original prompt by maximizing their edit distance where $\mathrm{Dist}(\cdot)$ is \textit{Levenshtein distance}~\citep{lcvenshtcin1966binary};
$\mathcal{L}_{\text{non-lang}}$ pushes the obfuscated prompt away from the natural language distribution by minimizing Shannon entropy: $H(\tilde{\rvx}) = - \sum_{c \in \mathcal{A}} v_{c} \log v_{c}$ where $\mathcal{A}$ is alphabet and $v_c$ is the frequency of character $c$ appearing in $\tilde{\rvx}$.
$\mathcal{L}_{\text{non-lang}}$ is motivated by the fact that natural language prompts are inherently portable.
By pushing the obfuscated prompt away from the natural language distribution, we make it dedicated for model-specific intricacies of the loss landscape, thereby implicitly minimizing its inter-model portability.

\begin{algorithm}
\caption{\ourshort{} algorithm.}
\label{alg:praglocker}
\begin{algorithmic}
\STATE {\bfseries Input:} Base LLM $p(\cdot)$, agent prompt $\rvx$, task training set $\mathcal{D}$, training steps $T$, noise set $\mathcal{N}$, loss function $\mathcal{L}(\cdot)$
\STATE {\bfseries Output:} Obfuscated agent prompt $\tilde{\rvx}$
\STATE $\tilde{\rvx}_0 \gets \mathrm{Init}(\rvx)$ \myalgcomment{Initialization}
\STATE $t \gets 0$
\WHILE{$t < T$}
    \STATE $(\rvq_t, \rvy_t) \sim \mathcal{D}$
    \STATE $\rvn_t \sim \mathcal{N}$ \myalgcomment{Sample character-level noise}
    \STATE $\tilde{\rvx}_{t+1}' \gets \tilde{\rvx}_t + \rvn_t$ \myalgcomment{In-place noise injection}
    \STATE $l_t \gets \mathcal{L}(p(\rvy_t \vert \rvq_t, \tilde{\rvx}_t), \rvy_t)$
    \STATE $l_t' \gets \mathcal{L}(p(\rvy_t \vert \rvq_t, \tilde{\rvx}_{t+1}'), \rvy_t)$
    \IF{$l_t' < l_t$}
        \STATE $\tilde{\rvx}_{t+1} \gets \tilde{\rvx}_{t+1}'$ \myalgcomment{Accept update if loss is lower}
    \ELSE
        \STATE $\tilde{\rvx}_{t+1} \gets \tilde{\rvx}_{t}$ \myalgcomment{Discard update otherwise}
    \ENDIF
    \STATE $t \gets t + 1$
\ENDWHILE
\end{algorithmic}
\end{algorithm}


\begin{table*}[t]
\renewcommand{\arraystretch}{0.8}
\centering
\setlength{\tabcolsep}{2pt}
\caption{Measuring prompt \textbf{non-portability} across different underlying LLMs.
For each \emph{Target LLM}, we develop a prompt and then simulate prompt theft by evaluating its performance when copied and reused on other LLMs. 
We report performance for each agent--task pair under four settings: \emph{Without protection}, \emph{PragLocker (ours)}, and two variants \emph{\ourtune{}} and \emph{\ourcode{}}. We report the \textit{Mean Portability Loss}, which is calculated by summing all metrics and expressing it as a multiple of without protection.
}
\resizebox{2.07\columnwidth}{!}{
\begin{tabular}{cccccccccccccccccc}
\toprule
\multirow{2}{*}{\textbf{Agent}} & \multirow{2}{*}{\textbf{Tasks}} & \multirow{2}{*}{\textbf{Target LLM}} & \multicolumn{3}{c}{\textbf{Without protection}} & & \multicolumn{3}{c}{\textbf{PragLocker (Ours)}} & & \multicolumn{3}{c}{\textbf{\ourtune{}}} & & \multicolumn{3}{c}{\textbf{\ourcode{}}} \\ 
\cline{4-6} \cline{8-10} \cline{12-14} \cline{16-18} \noalign{\smallskip} & & & GPT-4o & Gemini 2 & DeepSeek & & GPT-4o & Gemini 2 & DeepSeek & &  GPT-4o & Gemini 2 & DeepSeek & &  GPT-4o & Gemini 2 & DeepSeek \\
\cline{1-18} \noalign{\smallskip}

\multirow{6}{*}{LessonL} 
& \multirow{3}{*}{HumanEval} 
  & GPT-4o     & - & 98.78 & 97.56 &  &- & 3.04 & 1.22 & & - & 92.07 & 90.85& & - & 95.73 & 94.51 \\
& & Gemini 2   & 93.90 & - & 97.56 & & 0.61 & - & 1.22 & & 85.36 & - & 89.63 & & 92.07 & - & 95.12 \\
& & DeepSeek   & 93.90 & 98.78 & - & & 0.61 & 2.44 & - & & 82.93 & 93.29 & - & & 90.24 & 95.73 & - \\
\cline{2-18} \noalign{\smallskip}

& \multirow{3}{*}{MBPP} 
  & GPT-4o     & - & 97.33 & 94.56 &  &- & 1.03 & 0.72 & & - & 87.26 & 85.21& & - & 96.50 & 93.94 \\
& & Gemini 2   & 91.89 & - & 94.56 & & 0.51 & - & 0.62 & & 81.21 & - & 84.29 & & 91.37 & - & 93.43 \\
& & DeepSeek   & 91.89 & 97.33 & - & & 0.62 & 0.92 & - & & 80.59 & 87.78 & - & & 91.47 & 96.71 & - \\
\cline{1-18} \noalign{\smallskip}

\multirow{6}{*}{ReadAgent}
& \multirow{3}{*}{NarrativeQA}
  & GPT-4o     & - & 20.81 & 24.22 &  &- & 9.89 & 10.16 & & - & 12.83 & 12.43& & - & 22.61 & 22.74 \\
& & Gemini 2   & 23.52 & - & 24.22 & & 10.11 & - & 10.13 & & 11.46 & - & 12.32 & & 21.87 & - & 23.50 \\
& & DeepSeek   & 23.52 & 20.81 & - & & 9.14 & 8.07 & - & & 11.10 & 11.95 & - & & 21.17 & 23.72 & - \\
\cline{2-18} \noalign{\smallskip}

& \multirow{3}{*}{QUALITY}
  & GPT-4o     & - & 80.90 & 75.90 &  &- & 50.68 & 53.41 & & - & 65.95 & 65.86& & - & 80.23 & 84.96 \\
& & Gemini 2   & 86.19 & - & 88.06 & & 52.20 & - & 54.44 & & 62.63 & - & 64.54 & & 85.31 & - & 87.11 \\
& & DeepSeek   & 86.19 & 75.90 & - & & 51.12 & 48.56 & - & & 63.03 & 66.12 & - & & 84.15 & 80.50 & - \\
\cline{1-18} \noalign{\smallskip}

\multirow{6}{*}{A-Mem}
& \multirow{3}{*}{LoCoMo}
  & GPT-4o     & - & 22.85 & 25.62 &  &- & 0.08 & 0.12 & & - & 15.76 & 16.13& & - & 21.73 & 23.95 \\
& & Gemini 2   & 24.61 & - & 25.62 & & 0.06 & - & 0.14 & & 11.98 & - & 15.22 & & 24.03 & - & 25.48 \\
& & DeepSeek   & 24.61 & 22.85 & - & & 0.06 & 0.10 & - & & 12.87 & 16.61 & - & & 24.70 & 22.46 & - \\
\cline{2-18} \noalign{\smallskip}

& \multirow{3}{*}{DialSim}
  & GPT-4o     & - & 2.17 & 3.62 &  &- & 0.04 & 0.10 & & - & 1.28 & 2.08& & - & 2.08 & 3.41 \\
& & Gemini 2   & 2.86 & - & 3.62 & & 0.07 & - & 0.11 & & 1.45 & - & 2.59 & & 2.46 & - & 3.53 \\
& & DeepSeek   & 2.86 & 2.17 & - & & 0.06 & 0.04 & - & & 1.57 & 1.36 & - & & 2.79 & 2.15 & - \\
\cline{1-18} \noalign{\smallskip}

\rowcolor{gray!20}
\multicolumn{3}{c}{\textbf{Cross-LLM Portability Ratio ($\downarrow$)}} & \multicolumn{3}{c}{\textbf{1.00}$\times$} &  & \multicolumn{3}{c}{\textbf{0.20}$\times$} &  & \multicolumn{3}{c}{\textbf{0.82}$\times$} &  & \multicolumn{3}{c}{\textbf{0.99}$\times$} \\
\bottomrule
\end{tabular}
}
\label{Main_Result}
\end{table*}

\section{Experiments}

\subsection{Experimental Settings}
\paragraph{\textbf{Agents.}}
We evaluate PragLocker across diverse agent domains by considering three representative agents with different tasks and designs: LessonL~\cite{liu2025lessons} (multi-agent programming with a shared lesson bank), ReadAgent~\cite{lee2024human} (long-context reading with episodic memory and gist-level compression), and A-MEM~\cite{xu2025mem} (long-term memory via structured notes and dynamic indexing). Each agent is instantiated with three proprietary backbone LLMs: GPT-4o (GPT-4o)~\cite{openai_gpt4o_2024}, Gemini 2 Flash Preview (Gemini 2)~\cite{google_gemini_api_models}, and DeepSeek Chat (DeepSeek)~\cite{deepseek_chat_web}.


\mypar{\textbf{Dataset and Metrics.}}
We evaluate each agent on two established benchmarks and report their standard metrics: LessonL uses HumanEval~\cite{chen2021evaluating} and MBPP~\cite{austin2021program} (pass@1); ReadAgent uses NarrativeQA~\cite{kovcisky2018narrativeqa} (token-level F1) and QuALITY~\cite{pang2022quality} (accuracy); A-MEM uses LoCoMo~\cite{maharana2024evaluating} and DialSim~\cite{kim2024dialsim} (token-level F1, open-ended QA).

\mypar{\textbf{Baselines.}}
We compare against two controlled baselines (PragLocker ablations), since existing prompt protection/obfuscation methods rely on assumptions mismatched to our black-box untrusted-deployment setting. Specifically, \textit{\ourtune{}} performs only the optimization stage, skipping the initialization transformation; \textit{\ourcode{}} applies only the transformation that rewrites the prompt into a code-like representation, without subsequent optimization.


\subsection{Main Results} \label{subsec:main_results}

\paragraph{\textbf{Non-Portability.}}
We first evaluate whether PragLocker resists \emph{cross-model misuse} after prompt theft. We simulate an attacker who obtains the deployed (obfuscated) prompt and reuses it on other LLMs. As shown in Table~\ref{Main_Result}, prompts are highly portable \emph{without protection}: copying a prompt to a different underlying LLM preserves strong performance (e.g., a HumanEval prompt achieves 98.78 on Gemini~2 and 97.56 on DeepSeek). In contrast, \textit{PragLocker makes stolen prompts largely unusable}, driving performance close to zero across agents and tasks (e.g., 3.04/1.22), and achieves the lowest relative mean portability (0.2$\times$), compared to 1$\times$ without protection. Moreover, full PragLocker outperforms the simplified baselines: \ourcode{} remains nearly as portable as no protection (0.99$\times$), whereas \ourtune{} only partially reduces portability (0.82$\times$).

\begin{table}[t]
\renewcommand{\arraystretch}{0.6}
\centering
\setlength{\tabcolsep}{2pt}
\caption{
\textbf{Performance preservation} of PragLocker-protected prompts.
We compare the original prompt with its PragLocker-protected counterpart on the target LLM.
}
\resizebox{1.0\columnwidth}{!}{
\begin{tabular}{ccccccccc}
\toprule
\multirow{2}{*}{\textbf{Agent}} &
\multirow{2}{*}{\textbf{Tasks}} &
\multicolumn{3}{c}{\textbf{Without protection}} &
& 
\multicolumn{3}{c}{\textbf{After protection}} \\
\cline{3-5} \cline{7-9} \noalign{\smallskip}
& &
GPT-4o & Gemini 2 & DeepSeek &
& 
GPT-4o & Gemini 2 & DeepSeek \\
\cline{1-9} \noalign{\smallskip}

\multirow{2}{*}{LessonL}
& HumanEval
  & 93.90 & 98.78 &97.56
  &  
  & 94.51 & 99.39 & 98.17 \\
\cline{2-9} \noalign{\smallskip}

& MBPP
  & 91.89 & 97.33 & 94.56
  &  
  & 91.99 & 97.54 & 94.46 \\
\cline{1-9} \noalign{\smallskip}

\multirow{2}{*}{ReadAgent}
& NarrativeQA
  & 23.52 & 20.81 & 24.22
  &  
  & 23.61 & 20.93 & 25.81 \\
\cline{2-9} \noalign{\smallskip}

& QUALITY
  & 86.19 & 75.90 & 88.06 &  
  & 86.03 & 76.13 & 87.97 \\
\cline{1-9} \noalign{\smallskip}

\multirow{2}{*}{A-Mem}
& LoCoMo
  & 24.61 & 22.85 & 25.62
  &  
  & 25.11 & 23.01 & 26.80 \\
\cline{2-9} \noalign{\smallskip}

& DialSim
  & 2.86 & 2.17 & 3.62
  &  
  & 2.84 & 2.31 & 3.70 \\
\cline{1-9} \noalign{\smallskip}
\rowcolor{gray!20}
\multicolumn{2}{c}{\textbf{Performance Preservation ($\uparrow$)}} & \multicolumn{3}{c}{\textbf{1.00}$\times$} &  & \multicolumn{3}{c}{\textbf{1.01}$\times$} \\
\bottomrule
\end{tabular}
}
\label{Performance_Preservation}
\vskip -0.1in
\end{table}


\begin{figure*}[t]
\centering
\begin{tcbraster}[raster columns=2,raster equal height,nobeforeafter]
\begin{tcolorbox}[
  colback=gray!10!white,
  colframe=black,
  left=2mm,right=2mm,top=1mm,bottom=1mm,
  boxsep=0pt,
]
\footnotesize
\tcbsubtitle[colback=gray,colframe=gray]{Original Prompt}
\vspace{-2.5mm}
\begin{lstlisting}[style=promptxml,
  basicstyle=\ttfamily\scriptsize,
]
Given the following question, generate several keywords, using ',' as the separator.(*@\hfill{\color{gray}\scriptsize\textit{// line 1}}@*)
Question: {How does photosynthesis work?}(*@\hfill{\color{gray}\scriptsize\textit{// line 2}}@*)
Format your response as a JSON object with a "keywords" field containing the selected text.(*@\hfill{\color{gray}\scriptsize\textit{// line 3}}@*)
Example response format:(*@\hfill{\color{gray}\scriptsize\textit{// line 4}}@*)
{{"keywords": "keyword1, keyword2, keyword3"}}(*@\hfill{\color{gray}\scriptsize\textit{// line 5}}@*)
\end{lstlisting}
\vspace{9.4mm}
\tcbsubtitle[colback=gray,colframe=gray]{DeepSeek Output}
\vspace{-2.5mm}
\begin{lstlisting}[style=promptxml,
  basicstyle=\ttfamily\scriptsize,
]
{"keywords": "photosynthesis, process, plants, sunlight, carbon dioxide, water, glucose, oxygen, chlorophyll, chloroplasts"}
\end{lstlisting}
\vspace{-2mm}
\tcbsubtitle[colback=gray,colframe=gray]{GPT-4o Output}
\vspace{-2.5mm}
\begin{lstlisting}[style=promptxml,
  basicstyle=\ttfamily\scriptsize,
]
{"keywords": "photosynthesis, process, light energy, chlorophyll, carbon dioxide, water, glucose, oxygen, plants"}
\end{lstlisting}
\vspace{-3mm}
\end{tcolorbox}
\begin{tcolorbox}[
  colback=gray!10!white,
  colframe=black,
  left=2mm,right=2mm,top=1mm,bottom=1mm,
  boxsep=0pt,
]
\footnotesize
\tcbsubtitle[colback=gray,colframe=gray]{PragLocker Protected Prompt (DeepSeek as target model)}
\vspace{-2.5mm}
\begin{lstlisting}[style=promptxml,
  basicstyle=\ttfamily\scriptsize,
]
<γ-7δS4erator="lexosy6thmnzeparatorI"cosmy*X>+<input t3pid"quer{" val="{How does photosynthesis work?}"/>_<<r6cesGOaode="keyword-exzract"vcout="$ynQeiw" scoring="tf4gf"/> (*@\hfill{\color{gray}\scriptsize\textit{// line 1}}@*)
<Jutput Vo%0qt="$son"'> (*@\hfill{\color{gray}\scriptsize\textit{// line 2}}@*)
<:/eld name="k*ywor$s< type="CBncat<nate9" separatorM",/" VEndVnq="d,rect"/>(*@\hfill{\color{gray}\scriptsize\textit{// line 3}}@*)
</%utput>(*@\hfill{\color{gray}\scriptsize\textit{// line 4}}@*)
<exa.p;e-template>(*@\hfill{\color{gray}\scriptsize\textit{// line 5}}@*)
<json>{"keywo6Rs": "α1, α2, α3"O</json>(*@\hfill{\color{gray}\scriptsize\textit{// line 6}}@*)
S/exaKple-template>(*@\hfill{\color{gray}\scriptsize\textit{// line 7}}@*)
B/γ,7δ>x x x -H< x xjxrx|x(*@\hfill{\color{gray}\scriptsize\textit{// line 8}}@*)
\end{lstlisting}
\vspace{-2mm}
\tcbsubtitle[colback=gray,colframe=gray]{DeepSeek Output}
\vspace{-2.5mm}
\begin{lstlisting}[style=promptxml,
  basicstyle=\ttfamily\scriptsize,
]
{"keywords": "photosynthesis, process, sunlight, chlorophyll, carbon dioxide, water, oxygen, glucose, plants, energy"}
\end{lstlisting}
\vspace{-2mm}
\tcbsubtitle[colback=gray,colframe=gray]{GPT-4o Output}
\vspace{-2.5mm}
\begin{lstlisting}[style=promptxml,
  basicstyle=\ttfamily\scriptsize,
]
{It looks like your input includes a mix of structured markup and placeholder-like syntax....}
\end{lstlisting}
\vspace{-3mm}
\end{tcolorbox}
\end{tcbraster}
\caption{Case study: the original prompt and the protected prompt on DeepSeek (target model) vs. GPT-4o; task: keyword extraction.}
\vskip -0.1in
\end{figure*}

\mypar{\textbf{Performance Preservation.}}
We further assess whether PragLocker preserves the agent’s utility by replacing the original system prompt with the protected prompt and re-running the same tasks. As Table~\ref{Performance_Preservation} shows negligible degradation: protected performance closely matches the original across agents/tasks (relative mean 1.01$\times$ vs.\ 1$\times$). \textit{The slight gain likely stems from the feedback-guided tuning stage, which can refine redundant or suboptimal phrasing in the original prompt while maintaining its intended functionality.}


\begin{table*}[t]
\renewcommand{\arraystretch}{1} 
\centering
\setlength{\tabcolsep}{2pt}
\caption{Comparison of inter-family and intra-family portability. Left: inter-family transfer. Right: intra-family transfer within the Qwen2.5 Instruct family (Qwen2.5-7B-Instruct (Qwen-7B) / Qwen2.5-14B-Instruct (-14B) / Qwen2.5-14B-Instruct-bnb-4bit (-14B-4bit)).}
\label{tab:portability_one_table_fixed}

\resizebox{1.5\columnwidth}{!}{
\begin{tabular}{ccccccccccc}
\toprule

\multirow{2}{*}{\textbf{Tasks}} &
\multirow{2}{*}{\shortstack{\textbf{Target}\\\textbf{LLM}}} &
\multirow{2}{*}{\shortstack{\textbf{Original}\\\textbf{performance}}} &
\multicolumn{3}{c}{\textbf{Inter-family  portability}} &
\multirow{2}{*}{\shortstack{\textbf{Target}\\\textbf{LLM}}} &
\multirow{2}{*}{\shortstack{\textbf{Original}\\\textbf{performance}}} &
\multicolumn{3}{c}{\textbf{Intra-family portability}} \\
\cmidrule(lr){4-6}\cmidrule(lr){9-11}
& & & GPT-4o & Gemini 2 & DeepSeek
& & & Qwen-7B & -14B & -14B-4bit \\
\midrule

\multirow{3}{*}{\shortstack{NarrativeQA
}}
& GPT-4o   & 23.52 & -   & 9.89 & 10.16
& Qwen-7B  & 16.21 & -   & 8.12 & 8.97 \\
& Gemini 2 & 20.81 & 10.11 & -   & 10.13
& -14B     & 19.36 & 6.55 & -   & 7.12 \\
& DeepSeek & 24.22 & 9.14 & 8.07 & -
& -14B-4bit & 18.98 & 5.96 & 7.54 & - \\
\midrule

\multirow{3}{*}{QUALITY}
& GPT-4o   & 86.19 & -   & 50.68 & 53.41
& Qwen-7B  & 67.60 & -   & 44.57 & 48.92 \\
& Gemini 2 & 75.90 & 52.20 & -   & 54.44
& -14B     & 74.84 & 43.45 & -   & 46.75 \\
& DeepSeek & 88.06 & 51.12 & 48.56 & -
& -14B-4bit & 72.07 & 41.14 & 43.23 & - \\

\bottomrule
\end{tabular}
}
\label{Intra_Family_Portability}
\end{table*}

\subsection{Case Studies} \label{subsec:case_study}
We present a concise case study illustrating how PragLocker converts a human-readable system prompt into a protected prompt for a keyword-extraction subtask from ReadAgent. We provide an executable code implementation of this case study in the Supplementary Material. Specifically, taking DeepSeek as the target model, the protected prompt preserves utility: it follows the intended instruction and returns relevant keywords in the required JSON structure, matching the behavior of the original prompt.

However, directly transferring the same protected prompt to GPT-4o fails: the model interprets the obfuscated markup as malformed or noisy input. Moreover, the protected prompt is unreadable and structurally perturbed, leaving at most sparse, non-actionable fragments, which makes reconstructing the original prompt from the deployed artifact highly impractical. This challenge compounds for more complex tasks, where obfuscation becomes even more effective.

\begin{table*}[t]
\renewcommand{\arraystretch}{0.6}
\centering
\setlength{\tabcolsep}{2pt}
\caption{Adaptive-attack evaluation results. We report task performance under three adaptive attackers: (i) \emph{Deobfuscation}, which optimizes the obfuscated prompt with an inverted objective; (ii) \emph{LLM-assisted recovery}, which queries the target LLM to reconstruct the original prompt; and (iii) \emph{Naive self-prompting}, where the target LLM synthesizes a fresh prompt for itself using only task input--output pairs.}
\resizebox{2.07\columnwidth}{!}{
\begin{tabular}{cccccccccccccccccc}
\toprule
\multirow{2}{*}{\textbf{Agent}} & \multirow{2}{*}{\textbf{Tasks}} & \multirow{2}{*}{\textbf{Target LLM}} & \multicolumn{3}{c}{\textbf{PragLocker (Ours)}} & & \multicolumn{3}{c}{\textbf{LLM-assisted recovery}} & & \multicolumn{3}{c}{\textbf{Deobfuscation}} & & \multicolumn{3}{c}{\textbf{Naive prompt}} \\ 
\cline{4-6} \cline{8-10} \cline{12-14} \cline{16-18} \noalign{\smallskip} & & & GPT-4o & Gemini 2 & DeepSeek & & GPT-4o & Gemini 2 & DeepSeek & &  GPT-4o & Gemini 2 & DeepSeek & &  GPT-4o & Gemini 2 & DeepSeek \\
\cline{1-18} \noalign{\smallskip}
\multirow{6}{*}{LessonL} 
& \multirow{3}{*}{HumanEval} 
  & GPT-4o     
  & - & 3.04 & 1.22 
  &  & - & 86.59 & 87.20 
  &  & - & 3.04 & 1.83
  &  & - & 94.51 & 95.12 \\
& & Gemini 2   
  & 0.61 & - & 1.22
  &  & 82.32 & - & 85.37
  &  & 0.61 & - & 1.22
  &  & 90.85 & - & 93.90 \\
& & DeepSeek 
  & 0.61 & 2.44 & -
  &  & 84.15 & 87.02 & -
  &  & 1.22 & 2.44 & -
  &  & 91.68 & 95.73 & - \\
\cline{2-18} \noalign{\smallskip}

& \multirow{3}{*}{MBPP} 
  & GPT-4o     
  & - & 1.03 & 0.72
  &  & - & 83.16 & 84.29
  &  & - & 1.13 & 1.33
  &  & - & 90.24 & 90.45 \\
& & Gemini 2   
  & 0.51 & - & 0.62
  &  & 79.57 & - & 82.34
  &  & 0.72 & - & 0.62
  &  & 85.22 & - & 90.14 \\
& & DeepSeek 
  & 0.62 & 0.92 & -
  &  & 82.03 & 83.68 & -
  &  & 1.32 & 0.92 & -
  &  & 85.73 & 90.97 & - \\
\cline{1-18} \noalign{\smallskip}

\multirow{6}{*}{ReadAgent}
& \multirow{3}{*}{NarrativeQA}
  & GPT-4o     
  & - & 9.89 & 10.16
  &  & - & 5.46 & 6.87
  &  & - & 10.21 & 10.34
  &  & - & 17.38 & 21.32 \\
& & Gemini 2   
  & 10.11 & - & 10.13
  &  & 4.64 & - & 5.29
  &  & 10.26 & - & 10.36
  &  & 19.28 & - & 21.05 \\
& & DeepSeek 
  & 9.14 & 8.07 & -
  &  & 6.71 & 6.33 & -
  &  & 9.28 & 8.20 & -
  &  & 21.03 & 18.85 & - \\
\cline{2-18} \noalign{\smallskip}

& \multirow{3}{*}{QUALITY}
  & GPT-4o     
  & - & 50.68 & 53.41
  &  & - & 21.82 & 39.06
  &  & - & 51.13 & 54.20
  &  & - & 72.11 & 83.97 \\
& & Gemini 2   
  & 52.20 & - & 54.44
  &  & 31.82 & - & 34.57
  &  & 53.54 & - & 54.88
  &  & 82.34 & - & 82.87 \\
& & DeepSeek 
  & 51.12 & 48.56 & -
  &  & 38.56 & 26.79 & -
  &  & 51.36 & 51.22 & -
  &  & 84.29 & 72.61 & - \\
\cline{1-18} \noalign{\smallskip}

\multirow{6}{*}{A-Mem}
& \multirow{3}{*}{LoCoMo}
  & GPT-4o     
  & - & 0.08 & 0.12
  &  & - & 14.08 & 16.23
  &  & - & 0.39 & 0.47
  &  &  & 19.21 & 22.86 \\
& & Gemini 2   
  & 0.06 & - & 0.14
  &  & 12.89 & - & 14.94
  &  & 0.25 & - & 0.55
  &  & 18.23 & - & 20.47 \\
& & DeepSeek 
  & 0.06 & 0.10 & -
  &  & 16.43 & 14.57 & -
  &  & 0.30 & 0.37 & -
  &  & 22.69 & 20.33 &  \\
\cline{2-18} \noalign{\smallskip}

& \multirow{3}{*}{DialSim}
  & GPT-4o     
  & - & 0.04 & 0.10
  &  & - & 1.22 & 1.75
  &  & - & 0.06 & 0.15
  &  & - & 2.07 & 2.56 \\
& & Gemini 2   
  & 0.07 & - & 0.11
  &  & 1.41 & - & 1.56
  &  & 0.12 & - & 0.18
  &  & 2.13 & - & 2.29 \\
& & DeepSeek 
  & 0.06 & 0.04 & -
  &  & 1.33 & 1.31 & -
  &  & 0.10 & 0.05 & -
  &  & 2.45 & 2.04 & - \\
\cline{1-18} \noalign{\smallskip}
\rowcolor{gray!20}
\multicolumn{3}{c}{\textbf{Relative Attack Gain ($\downarrow$)}} & \multicolumn{3}{c}{\textbf{1}$\times$} &  & \multicolumn{3}{c}{\textbf{3.49}$\times$} &  & \multicolumn{3}{c}{\textbf{1.03}$\times$} &  & \multicolumn{3}{c}{\textbf{4.78}$\times$} \\
\bottomrule
\end{tabular}
}
\label{Adaptive_Attacks}
\vskip -0.1in
\end{table*}

\subsection{Further Analysis}
\paragraph{\textbf{Ablation Study.}}
We conduct a stage-wise ablation to isolate the contributions of \emph{code transformation} and \emph{noise-injected optimization} under the same evaluation protocol.
As shown in Table~\ref{Main_Result}, the two stages play distinct roles and are most effective when combined. \textit{\ourcode{}}, which only rewrites the prompt into a code-like representation, offers little protection against cross-model reuse: its mean portability degradation remains $0.99\times$ as code-form prompts are still broadly executable across other LLMs; however, \textbf{\textit{we interpret this stage as primarily providing a structured initialization that further expands the optimization search space for the subsequent tuning stage.}} Specifically, \textit{\ourtune{}} (optimization-only) reduces portability to $0.82\times$ but remains far weaker than full PragLocker ($0.2\times$). Overall, strong resistance requires \emph{both} stages: code transformation supplies an obfuscating scaffold and optimization room, while optimization ultimately converts it into a strongly non-portable prompt representation.

\mypar{Intra-Family Portability.}
\label{subsec:intra-family-portability}
As PragLocker is optimized from target-LLM I/O feedback, a natural concern is whether behaviorally similar sibling models permit intra-family reuse.
We therefore evaluate pairwise portability within the Qwen2.5 Instruct family (Qwen-7B / -14B / -14B-4bit).
Table~\ref{Intra_Family_Portability} indicates strong non-portability under both scale shifts and quantization shifts, and exhibits no qualitative difference from inter-family portability on the left side of the table: on NarrativeQA, although the original performance is 16.21 (Qwen-7B), 19.36 (-14B), and 18.98 (-14B-4bit), transferring protected prompts across siblings drops to only 5.96--8.97 (e.g., 7B$\rightarrow$14B: 8.12; 14B$\rightarrow$7B: 6.55; 14B$\rightarrow$14B-4bit: 7.12).
We attribute this robustness to our noise-injected optimization objective, which tightly couples the protected prompt to the target model’s idiosyncratic feedback dynamics, so that even modest distributional shifts (scaling or quantization) substantially degrade reuse.

\newcommand\FramedBox[3]{%
  \setlength\fboxsep{0pt}
  \fbox{\parbox[t][#1][c]{#2}{\centering\huge #3}}}

\begin{figure}[!t]
\centering
\includegraphics[width=0.85\linewidth]{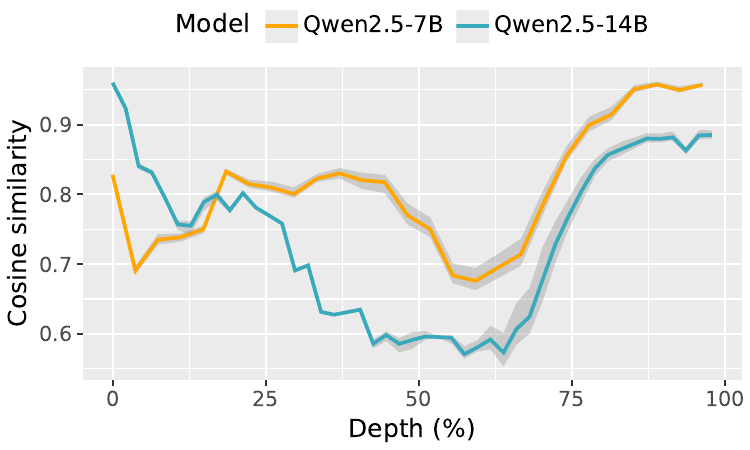}
\caption{Layerwise hidden state cosine similarity between obfuscated prompt and original prompt; non-portability direction: Qwen2.5-7B $\rightarrow$ Qwen2.5-14B; 95\% confidence interval is shown.}
\label{fig:mechanistic_analysis}
\vskip -0.1in
\end{figure}

\mypar{Hidden-State Alignment.}
To probe the mechanistic origin of non-portability, we compare how the prompt pair (original, obfuscated) is represented inside the target versus a non-target model. We use a protected prompt optimized on Qwen2.5-7B (target) and run the original and obfuscated prompts separately through the model. At each transformer layer $\ell$, we extract their hidden states and compute the cosine similarity between the resulting representations as a layerwise alignment score. We then feed the identical prompt pair into Qwen2.5-14B (non-target) to repeat the measurement.

\cref{fig:mechanistic_analysis} reveals a clear \emph{trajectory divergence} across depth. The two models behave similarly in early layers, but their alignment patterns separate as depth increases: after a mid-layer dip, Qwen2.5-7B exhibits a stronger recovery and reaches substantially higher similarity in the later layers, whereas Qwen2.5-14B remains less aligned. This suggests that PragLocker’s functionality preservation is realized through a target-specific internal pathway—its obfuscated prompt is mapped to representations that converge back toward those induced by the original prompt \emph{only} in the target model—providing mechanistic evidence that non-portability is formed predominantly in deeper representations rather than in shallow token-level processing.

\mypar{Token-Only Feedback.}
Our default implementation uses token log-probabilities to compute the cross-entropy task loss in Eq.~(1). To test whether this feedback is necessary, we evaluate a stricter token-only setting, where the optimizer only observes decoded outputs and uses the task metric as the black-box objective. Since this signal is sparser, we increase the optimization budget to 100 epochs.
As shown in Table~\ref{tab:token_only_readagent_quality}, PragLocker remains effective under token-only feedback with 100 optimization epochs. The optimized prompts preserve target utility, achieving 85.85\%, 71.31\%, and 82.05\% accuracy on GPT-4o, Gemini-2, and DeepSeek, respectively, while still showing clear cross-model degradation. For example, the GPT-4o optimized prompt drops from 85.85\% on GPT-4o to 58.16\% on Gemini-2 and 63.50\% on DeepSeek. These results show that log-probabilities are helpful but not required; PragLocker can operate with token-only black-box feedback at the cost of a larger optimization budget.

\begin{table}[t]
\centering
\small
\caption{Token-only feedback evaluation of ReadAgent on QuALITY. The prompt is optimized using only decoded token outputs, without access to token log-probabilities.}
\label{tab:token_only_readagent_quality}
\resizebox{1\columnwidth}{!}{
\begin{tabular}{cccccc}
\toprule
\textbf{Agent} & \textbf{Tasks} & \textbf{Target LLM} 
& \textbf{GPT-4o} 
& \makecell{\textbf{Gemini 2}} 
& \makecell{\textbf{DeepSeek}} \\
\midrule
\multirow{3}{*}{ReadAgent} 
& \multirow{3}{*}{QuALITY} 
& GPT-4o    & 85.85 & 58.16 & 63.50 \\
& 
& Gemini 2  & 59.67 & 71.31 & 60.82 \\
& 
& DeepSeek  & 58.56 & 55.24 & 82.05 \\
\bottomrule
\end{tabular}
}
\end{table}

\subsection{Adaptive Attacks} \label{subsec:adaptive_attack}
\paragraph{\textbf{LLM-Assisted Prompt Recovery.}}
In this section, we test whether the obfuscated prompt can be recovered at the \emph{text} level. Specifically, we simulate a sophisticated attacker who queries the \emph{target LLM itself} to (i) interpret the obfuscated prompt and reconstruct a natural-language instruction prompt, which is then (ii) substituted back into the agent system (more details are provided in \cref{subapp:llm_assisted_recovery_details}). 

Strikingly, as shown in Table \ref{Adaptive_Attacks}, \textbf{\textit{even the target LLM itself fails to reliably recover the underlying instructions}}, the recovered prompts remain unusable and do not restore transferable functionality. Specifically, LLM-assisted recovery yields only marginal gains over directly reusing the PragLocker-protected prompt (from $1\times$ to $3.49\times$ in relative mean performance).
This suggests that, after obfuscation, the protected prompt behaves less like a text-level semantic instruction and more like a \textbf{\emph{model-conditioned trigger}}: its utility is encoded in the target model’s idiosyncratic response dynamics, leaving little recoverable text-level semantics even to the target LLM itself.

\mypar{Deobfuscation Attack.}
We also consider a sophisticated attacker who knows our obfuscation pipeline and is given the same task instructions, input–output examples, and evaluation metrics as the defender. Under this setting, the attacker starts from the deployed obfuscated prompt $\tilde{x}$ and attempts to ``deobfuscate'' it by iteratively re-optimizing the prompt toward a more portable, human-readable instruction.
Concretely, the attacker drops $\mathcal{L}_{\text{dist}}$, keeps $\mathcal{L}_{\text{task}}$, flips the sign of $\mathcal{L}_{\text{non-lang}}$, and then iteratively tunes the prompt accordingly to push it back toward the natural-language distribution. For comparison, we also include a \emph{naive prompt} baseline, where the attacker provides the task query and ground-truth outputs and asks the target LLM to synthesize a fresh prompt for solving the task (details of the \textit{deobfuscation} and \textit{naive prompt} attacks are provided in \cref{subapp:deobfuscation_details}).

As shown in Table~\ref{Adaptive_Attacks}, this attack is largely ineffective: Relative Mean Performance improves by only $1.03\times$, and per-task metrics remain near unusable, far from approaching the unprotected prompt’s portability. Notably, it is even weaker than the naive prompt baseline ($4.78\times$), suggesting that \textit{attackers are better off directly eliciting a new prompt themself from the LLM than attempting to invert our obfuscation.}

\section{Limitation and Discussion}


\mypar{Tokenizer Compatibility.}
PragLocker injects character-level noise, which requires the target LLM’s tokenizer to have broad character coverage so that perturbed text still maps to valid subword/byte tokens (instead of unknown/rejected tokens). This assumption generally holds for mainstream production API LLMs (e.g., GPT-series~\cite{openai_chatgpt_2022}, Google Gemini~\cite{gemini_2023}).

\mypar{Prompt-Length Overhead.}
PragLocker can lengthen the system prompt and increase the prefill cost. In practice, the obfuscated prompt is a fixed prefix, so its KV states can be cached and reused across requests rather than recomputed each time. Thus, the extra cost is largely amortized per cached context, limiting its impact on end-to-end latency.

\mypar{Comparison with White-Box Defenses.}
PragLocker is designed for black-box prompt protection in proprietary LLM deployments, where the defender accesses the target model only through APIs. As a result, this work does not extensively compare against prompt-protection methods developed under stronger assumptions, such as white-box gradient access or model adaptation. These methods operate in a partially different regime and are therefore not directly comparable to our setting.

\mypar{Protection scope.}
PragLocker protects a specific and practically important component of agent IP: deployed prompts, which often encode task instructions, control logic, and format constraints in prompt-based agents. Its goal is to reduce the direct cross-model reusability of leaked prompts, rather than eliminate all routes to reproducing agent behavior. Broader behavior-reproduction attacks, such as imitation, prompt induction, or distillation from examples, follow a different attack path and are therefore largely orthogonal to PragLocker.



\section{Conclusions}
We study system-prompt theft in untrusted LLM-agent deployments, where a leaked prompt can be readily reused across models. Specifically, we identify prompts as the primary embodiment of agent IP and distill four requirements for protection. We propose PragLocker, a black-box protection method that transforms a prompt into a model-conditioned, non-portable form via structured initialization and noise-injected optimization. Experiments show that PragLocker strongly reduces portability to other LLMs, remaining robust to LLM-assisted recovery attempts.

\section*{Acknowledgements}
This work was supported by the Key R\&D Program of Ningbo under Grant No.2024Z115.

\section*{Impact Statement}

This paper presents work whose goal is to advance the field of machine learning. There are many potential societal consequences of our work, none of which we feel must be specifically highlighted here.

\balance

\bibliography{example_paper}

@inproceedings{yao2024promptcare,
  title={Promptcare: Prompt copyright protection by watermark injection and verification},
  author={Yao, Hongwei and Lou, Jian and Qin, Zhan and Ren, Kui},
  booktitle={2024 IEEE Symposium on Security and Privacy (SP)},
  pages={845--861},
  year={2024},
  organization={IEEE}
}

@article{yang2025promptcos,
  title={PromptCOS: Towards Content-only System Prompt Copyright Auditing for LLMs},
  author={Yang, Yuchen and Li, Yiming and Yao, Hongwei and Huang, Enhao and Shao, Shuo and Wang, Yuyi and Wang, Zhibo and Tao, Dacheng and Qin, Zhan},
  journal={arXiv preprint arXiv:2509.03117},
  year={2025}
}

@inproceedings{pape2025prompt,
  title={Prompt obfuscation for large language models},
  author={Pape, David and Mavali, Sina and Eisenhofer, Thorsten and Sch{\"o}nherr, Lea},
  booktitle={34th USENIX Security Symposium (USENIX Security 25)},
  pages={2323--2342},
  year={2025}
}

@article{karvandi2024reversing,
  title={The reversing machine: reconstructing memory assumptions},
  author={Karvandi, Mohammad Sina and Meghdadizanjani, Soroush and Arasteh, Sima and Monfared, Saleh Khalaj and Fallah, Mohammad K and Gorgin, Saeid and Lee, Jeong-A and van der Kouwe, Erik},
  journal={arXiv preprint arXiv:2405.00298},
  year={2024}
}

@article{costan2016intel,
  title={Intel SGX explained},
  author={Costan, Victor and Devadas, Srinivas},
  journal={Cryptology ePrint Archive},
  year={2016}
}

@online{kubernetes_encrypt_data_at_rest,
  title        = {Encrypting Confidential Data at Rest},
  author       = {{The Kubernetes Authors}},
  year         = {2025},
  month        = may,
  url          = {https://kubernetes.io/docs/tasks/administer-cluster/encrypt-data/},
  note         = {Last modified May 09, 2025. Accessed 2025-12-25}
}

@online{openai_api_reference_introduction,
  title        = {API Reference: Introduction},
  author       = {{OpenAI}},
  year         = {2025},
  url          = {https://platform.openai.com/docs/api-reference/introduction},
  note         = {Accessed 2025-12-25}
}

@online{google_gemini_api_models,
  title        = {Gemini models | Gemini API | Google AI for Developers},
  author       = {{Google}},
  year         = {2025},
  url          = {https://ai.google.dev/gemini-api/docs/models},
  note         = {Last updated 2025-12-18 (UTC). Accessed 2025-12-25}
}

@techreport{nvidia_hcc_whitepaper_wp11459,
  title        = {Confidential Compute on NVIDIA Hopper H100},
  author       = {Nertney, Rob},
  institution  = {NVIDIA},
  number       = {WP-11459-001},
  year         = {2023},
  month        = jul,
  date         = {2023-07-25},
  url          = {https://images.nvidia.com/aem-dam/en-zz/Solutions/data-center/HCC-Whitepaper-v1.0.pdf},
  note         = {Version 1.0. Accessed 2025-12-25}
}

@inproceedings{lin2025emojiprompt,
  title={Emojiprompt: Generative prompt obfuscation for privacy-preserving communication with cloud-based llms},
  author={Lin, Sam and Hua, Wenyue and Wang, Zhenting and Jin, Mingyu and Fan, Lizhou and Zhang, Yongfeng},
  booktitle={Proceedings of the 2025 Conference of the Nations of the Americas Chapter of the Association for Computational Linguistics: Human Language Technologies (Volume 1: Long Papers)},
  pages={12342--12361},
  year={2025}
}

@article{gim2024confidential,
  title={Confidential prompting: Protecting user prompts from cloud llm providers},
  author={Gim, In and Li, Caihua and Zhong, Lin},
  journal={arXiv preprint arXiv:2409.19134},
  year={2024}
}

@inproceedings{yuan2025scx,
  title={SCX: Stateless KV-Cache Encoding for Cloud-Scale Confidential Transformer Serving},
  author={Yuan, Mu and Zhang, Lan and Zeng, Liekang and Jiang, Siyang and Yang, Bufang and Duan, Di and Xing, Guoliang},
  booktitle={Proceedings of the ACM SIGCOMM 2025 Conference},
  pages={39--54},
  year={2025}
}

@inproceedings{su2025runtime,
  title={Runtime Attestation for Secure LLM Serving in Cloud-Native Trusted Execution Environments},
  author={Su, Jianchang and Zhang, Wei},
  booktitle={Machine Learning for Computer Architecture and Systems 2025}
}

@article{rana2023kms_survey,
  title   = {A comprehensive survey of cryptography key management systems},
  author  = {Rana, Subhabrata and Khoda Parast, Fatemeh and Kelly, Brett and Wang, Yang and Kent, Kenneth B.},
  journal = {Journal of Information Security and Applications},
  year    = {2023},
  volume  = {74},
  pages   = {103607},
  doi     = {10.1016/j.jisa.2023.103607}
}

@misc{cursor_features,
  title        = {Features · Cursor},
  author       = {{Cursor}},
  year         = {2025},
  url          = {https://cursor.com/features},
  urldate      = {2025-12-29}
}

@misc{friedman2021copilot,
  title        = {Introducing {GitHub} Copilot: your {AI} pair programmer},
  author       = {Nat Friedman},
  year         = {2021},
  month        = jun,
  url          = {https://github.com/features/copilot},
  note         = {Updated Feb 23, 2022},
  urldate      = {2025-12-29}
}

@misc{zapier_agents,
  title        = {Build {AI} teammates with Zapier Agents},
  author       = {{Zapier}},
  year         = {2025},
  url          = {https://zapier.com/agents},
  urldate      = {2025-12-29}
}

@misc{spector2025zapier_agents_guide,
  title        = {Zapier Agents: Work hand in hand with {AI} agents},
  author       = {Steph Spector},
  year         = {2025},
  month        = nov,
  url          = {https://zapier.com/blog/zapier-agents-guide/},
  note         = {Most recently updated Nov 2025},
  urldate      = {2025-12-29}
}

@misc{openai_text_guide_instructions,
  title        = {Text generation (OpenAI API Guide)},
  author       = {{OpenAI}},
  year         = {2025},
  url          = {https://platform.openai.com/docs/guides/text},
  urldate      = {2025-12-29}
}

@misc{anthropic_system_prompts,
  title        = {Giving Claude a role with a system prompt (Claude Docs)},
  author       = {{Anthropic}},
  year         = {2025},
  url          = {https://platform.claude.com/docs/en/build-with-claude/prompt-engineering/system-prompts},
  urldate      = {2025-12-29}
}

@misc{openai_chatgpt_2022,
  title        = {Introducing ChatGPT},
  author       = {{OpenAI}},
  year         = {2022},
  month        = nov,
  url          = {https://openai.com/index/chatgpt/},
  urldate      = {2025-12-29}
}

@article{gemini_2023,
  title={Gemini: a family of highly capable multimodal models},
  author={Team, Gemini and Anil, Rohan and Borgeaud, Sebastian and Alayrac, Jean-Baptiste and Yu, Jiahui and Soricut, Radu and Schalkwyk, Johan and Dai, Andrew M and Hauth, Anja and Millican, Katie and others},
  journal={arXiv preprint arXiv:2312.11805},
  year={2023}
}

@inproceedings{yang2025prsa,
  title={$\{$PRSA$\}$: Prompt stealing attacks against $\{$Real-World$\}$ prompt services},
  author={Yang, Yong and Li, Changjiang and Li, Qingming and Ma, Oubo and Wang, Haoyu and Wang, Zonghui and Gao, Yandong and Chen, Wenzhi and Ji, Shouling},
  booktitle={34th USENIX security symposium (USENIX Security 25)},
  pages={2283--2302},
  year={2025}
}

@article{sahoo2024prompt_survey,
  title        = {A Systematic Survey of Prompt Engineering in Large Language Models: Techniques and Applications},
  author       = {Pranab Sahoo and Ayush Kumar Singh and Sriparna Saha and Vinija Jain and Samrat Mondal and Aman Chadha},
  journal      = {arXiv preprint arXiv:2402.07927},
  year         = {2024},
  url          = {https://arxiv.org/abs/2402.07927}
}

@article{nist_sp800_145,
  title={The nist definition of cloud computing},
  author={Cloud, Hybrid},
  journal={National institute of science and technology, special publication},
  volume={800},
  number={2011},
  pages={145},
  year={2011}
}

@inproceedings{hui2024pleak,
  title={Pleak: Prompt leaking attacks against large language model applications},
  author={Hui, Bo and Yuan, Haolin and Gong, Neil and Burlina, Philippe and Cao, Yinzhi},
  booktitle={Proceedings of the 2024 on ACM SIGSAC Conference on Computer and Communications Security},
  pages={3600--3614},
  year={2024}
}

@inproceedings{wang2024raccoon,
  title        = {Raccoon: Prompt Extraction Benchmark of LLM-Integrated Applications},
  author       = {Junlin Wang and Tianyi Yang and Roy Xie and Bhuwan Dhingra},
  booktitle    = {Findings of the Association for Computational Linguistics: ACL 2024},
  year         = {2024},
  pages        = {13349--13365},
  url          = {https://aclanthology.org/2024.findings-acl.791/}
}

@article{wang2024survey_llm_agents,
  title   = {A survey on large language model based autonomous agents},
  author  = {Wang, Lei and Ma, Chen and Feng, Xueyang and Zhang, Zeyu and Yang, Hao and Zhang, Jingsen and Chen, Zhiyuan and Tang, Jiakai and Chen, Xu and Lin, Yankai and Zhao, Wayne Xin and Wei, Zhewei and Wen, Ji-Rong},
  journal = {Frontiers of Computer Science},
  year    = {2024},
  doi     = {10.1007/s11704-024-40231-1}
}

@misc{ManusAI_Website_2026,
  title        = {Manus: Hands On AI},
  author       = {{Manus AI}},
  year         = {2026},
  howpublished = {\url{https://manus.im/}},
  note         = {Accessed: 2026-01-08}
}

@inproceedings{mcgovern2025fingerprints,
  title     = {Your Large Language Models Are Leaving Fingerprints},
  author    = {McGovern, H. E. and Stureborg, R. and Suhara, Y. and others},
  booktitle = {Proceedings of the 1st Workshop on GenAI Content Detection (GenAIDetect)},
  year      = {2025},
  pages     = {85--95}
}

@article{rastrigin1963convergence,
  title={The convergence of the random search method in the extremal control of a many parameter system},
  author={Rastrigin, LA},
  journal={Automaton \& Remote Control},
  volume={24},
  pages={1337--1342},
  year={1963}
}

@inproceedings{andriushchenko2025jailbreaking,
title={Jailbreaking Leading Safety-Aligned {LLM}s with Simple Adaptive Attacks},
author={Maksym Andriushchenko and Francesco Croce and Nicolas Flammarion},
booktitle={The Thirteenth International Conference on Learning Representations},
year={2025},
url={https://openreview.net/forum?id=hXA8wqRdyV}
}

@article{lee2024human,
  title={A human-inspired reading agent with gist memory of very long contexts},
  author={Lee, Kuang-Huei and Chen, Xinyun and Furuta, Hiroki and Canny, John and Fischer, Ian},
  journal={arXiv preprint arXiv:2402.09727},
  year={2024}
}

@article{xu2025mem,
  title={A-mem: Agentic memory for llm agents},
  author={Xu, Wujiang and Liang, Zujie and Mei, Kai and Gao, Hang and Tan, Juntao and Zhang, Yongfeng},
  journal={arXiv preprint arXiv:2502.12110},
  year={2025}
}

@article{liu2025lessons,
  title={Lessons Learned: A Multi-Agent Framework for Code LLMs to Learn and Improve},
  author={Liu, Yuanzhe and Deng, Ryan and Kaler, Tim and Chen, Xuhao and Leiserson, Charles E and Ma, Yao and Chen, Jie},
  journal={arXiv preprint arXiv:2505.23946},
  year={2025}
}

@article{chen2021evaluating,
  title={Evaluating large language models trained on code},
  author={Chen, Mark},
  journal={arXiv preprint arXiv:2107.03374},
  year={2021}
}

@article{austin2021program,
  title={Program synthesis with large language models},
  author={Austin, Jacob and Odena, Augustus and Nye, Maxwell and Bosma, Maarten and Michalewski, Henryk and Dohan, David and Jiang, Ellen and Cai, Carrie and Terry, Michael and Le, Quoc and others},
  journal={arXiv preprint arXiv:2108.07732},
  year={2021}
}

@article{kovcisky2018narrativeqa,
  title={The narrativeqa reading comprehension challenge},
  author={Ko{\v{c}}isk{\`y}, Tom{\'a}{\v{s}} and Schwarz, Jonathan and Blunsom, Phil and Dyer, Chris and Hermann, Karl Moritz and Melis, G{\'a}bor and Grefenstette, Edward},
  journal={Transactions of the Association for Computational Linguistics},
  volume={6},
  pages={317--328},
  year={2018},
  publisher={MIT Press One Rogers Street, Cambridge, MA 02142-1209, USA journals-info~…}
}

@inproceedings{pang2022quality,
  title={QuALITY: Question answering with long input texts, yes!},
  author={Pang, Richard Yuanzhe and Parrish, Alicia and Joshi, Nitish and Nangia, Nikita and Phang, Jason and Chen, Angelica and Padmakumar, Vishakh and Ma, Johnny and Thompson, Jana and He, He and others},
  booktitle={Proceedings of the 2022 Conference of the North American Chapter of the Association for Computational Linguistics: Human Language Technologies},
  pages={5336--5358},
  year={2022}
}

@article{maharana2024evaluating,
  title={Evaluating very long-term conversational memory of llm agents},
  author={Maharana, Adyasha and Lee, Dong-Ho and Tulyakov, Sergey and Bansal, Mohit and Barbieri, Francesco and Fang, Yuwei},
  journal={arXiv preprint arXiv:2402.17753},
  year={2024}
}

@article{kim2024dialsim,
  title={DialSim: A Real-Time Simulator for Evaluating Long-Term Multi-Party Dialogue Understanding of Conversation Systems},
  author={Kim, Jiho and Chay, Woosog and Hwang, Hyeonji and Kyung, Daeun and Chung, Hyunseung and Cho, Eunbyeol and Jo, Yohan and Choi, Edward},
  journal={arXiv preprint arXiv:2406.13144},
  year={2024}
}

@article{chang2024efficient,
  title={Efficient prompting methods for large language models: A survey},
  author={Chang, Kaiyan and Xu, Songcheng and Wang, Chenglong and Luo, Yingfeng and Liu, Xiaoqian and Xiao, Tong and Zhu, Jingbo},
  journal={arXiv preprint arXiv:2404.01077},
  year={2024}
}

@inproceedings{li2025prompt,
  title={Prompt compression for large language models: A survey},
  author={Li, Zongqian and Liu, Yinhong and Su, Yixuan and Collier, Nigel},
  booktitle={Proceedings of the 2025 Conference of the Nations of the Americas Chapter of the Association for Computational Linguistics: Human Language Technologies (Volume 1: Long Papers)},
  pages={7182--7195},
  year={2025}
}

@article{shin2020autoprompt,
  title={Autoprompt: Eliciting knowledge from language models with automatically generated prompts},
  author={Shin, Taylor and Razeghi, Yasaman and Logan IV, Robert L and Wallace, Eric and Singh, Sameer},
  journal={arXiv preprint arXiv:2010.15980},
  year={2020}
}

@inproceedings{singh2023explaining,
  title={Explaining data patterns in natural language with language models},
  author={Singh, Chandan and Morris, John X and Aneja, Jyoti and Rush, Alexander M and Gao, Jianfeng},
  booktitle={Proceedings of the 6th BlackboxNLP Workshop: Analyzing and Interpreting Neural Networks for NLP},
  pages={31--55},
  year={2023}
}

@inproceedings{deng2022rlprompt,
  title={Rlprompt: Optimizing discrete text prompts with reinforcement learning},
  author={Deng, Mingkai and Wang, Jianyu and Hsieh, Cheng-Ping and Wang, Yihan and Guo, Han and Shu, Tianmin and Song, Meng and Xing, Eric and Hu, Zhiting},
  booktitle={Proceedings of the 2022 conference on empirical methods in natural language processing},
  pages={3369--3391},
  year={2022}
}

@incollection{schwefel1977evolutionsstrategien,
  title={Evolutionsstrategien f{\"u}r die numerische optimierung},
  author={Schwefel, Hans-Paul},
  booktitle={Numerische Optimierung von Computer-Modellen Mittels der Evolutionsstrategie: Mit Einer Vergleichenden Einf{\"u}hrung in Die Hill-Climbing-und Zufallsstrategie},
  pages={123--176},
  year={1977},
  publisher={Springer}
}

@article{salimans2017evolution,
  title={Evolution strategies as a scalable alternative to reinforcement learning},
  author={Salimans, Tim and Ho, Jonathan and Chen, Xi and Sidor, Szymon and Sutskever, Ilya},
  journal={arXiv preprint arXiv:1703.03864},
  year={2017}
}

@inproceedings{lcvenshtcin1966binary,
  title={Binary coors capable or ‘correcting deletions, insertions, and reversals},
  author={Lcvenshtcin, VI},
  booktitle={Soviet physics-doklady},
  volume={10},
  number={8},
  year={1966}
}

@article{nikolaou2025language,
  title={Language models are injective and hence invertible},
  author={Nikolaou, Giorgos and Mencattini, Tommaso and Crisostomi, Donato and Santilli, Andrea and Panagakis, Yannis and Rodol{\`a}, Emanuele},
  journal={arXiv preprint arXiv:2510.15511},
  year={2025}
}

@article{lester2021power,
  title={The power of scale for parameter-efficient prompt tuning},
  author={Lester, Brian and Al-Rfou, Rami and Constant, Noah},
  journal={arXiv preprint arXiv:2104.08691},
  year={2021}
}

@article{zhang2024agent,
  title={Agent-pro: Learning to evolve via policy-level reflection and optimization},
  author={Zhang, Wenqi and Tang, Ke and Wu, Hai and Wang, Mengna and Shen, Yongliang and Hou, Guiyang and Tan, Zeqi and Li, Peng and Zhuang, Yueting and Lu, Weiming},
  journal={arXiv preprint arXiv:2402.17574},
  year={2024}
}

@misc{openai_gpt4o_2024,
  author       = {{OpenAI}},
  title        = {Hello GPT-4o},
  year         = {2024},
  month        = may,
  url          = {https://openai.com/index/hello-gpt-4o/},
  note         = {Accessed: 2026-01-28}
}

@misc{deepseek_chat_web,
  author       = {{DeepSeek}},
  title        = {DeepSeek Chat},
  year         = {n.d.},
  url          = {https://chat.deepseek.com/},
  note         = {Accessed: 2026-01-28}
}
\bibliographystyle{icml2026}

\newpage
\appendix
\onecolumn

\section*{Appendix}

\begin{itemize}[itemsep=-1pt,label={$\blacktriangleright$}]
  \item \textbf{\hyperref[app:theoretical_motivation]{A} \hyperref[app:theoretical_motivation]{Discussions on Theoretical Motivation for \ourshort{}}} \dotfill \pageref{app:theoretical_motivation}
  \item \textbf{\hyperref[app:main_experiment_details]{B} \hyperref[app:main_experiment_details]{Details on Main Experiment}} \dotfill \pageref{app:main_experiment_details}
  \item \textbf{\hyperref[app:details-on-case-study]{C} \hyperref[app:details-on-case-study]{Details on Case Study}} \dotfill \pageref{app:details-on-case-study}
  \item \textbf{\hyperref[app:adaptive_attack_details]{D} \hyperref[app:adaptive_attack_details]{Details on Adaptive Attacks}} \dotfill \pageref{app:adaptive_attack_details}
    \begin{itemize}[label={$\rhd$}, topsep=-1pt, itemsep=-1pt]
      \item \hyperref[subapp:llm_assisted_recovery_details]{D.1} \hyperref[subapp:llm_assisted_recovery_details]{LLM-Assisted Prompt Recovery.} \dotfill \pageref{subapp:llm_assisted_recovery_details}
      \item \hyperref[subapp:deobfuscation_details]{D.2} \hyperref[subapp:deobfuscation_details]{Deobfuscation Attack.} \dotfill \pageref{subapp:deobfuscation_details}
    \end{itemize}
  \item \textbf{\hyperref[app:discussions-on-methodology]{E} \hyperref[app:discussions-on-methodology]{Discussions on \ourshort{} Methodology}} \dotfill \pageref{app:discussions-on-methodology}
\end{itemize}
\clearpage

\begin{table}[t]
\caption{Glossary.}
\label{tab:glossary}
\centering
\setlength{\tabcolsep}{4pt}
\footnotesize
\begin{tabular}{ll}
\toprule
Symbol & Meaning \\
\midrule
$d \in \R$ & Model dimension. \\
$\mathcal{V}$ & Model vocabulary. \\
$\mathcal{P} \subset \{ \mathcal{V}, \mathcal{V}^2, \dots \}$ & Prompt space. \\
$\rvx \in \mathcal{P}$ & Original agent prompt. \\
$\tilde{\rvx} \in \mathcal{P}$ & Obfuscated prompt for target prompt $\rvx$. \\
$\mathrm{Embed}(\cdot)$ & Embedding layer. \\
$\rvh \in \R^d$ & Model embeddings; $\rvh = \mathrm{Embed}(\rvx)$. \\
$\bm{\delta}$ & Perturbation to embeddings. \\
$\rvq$ & An instance of task query, often used together with the agent prompt. \\
$\mathcal{Q}$ & Set of task queries. \\
$f(\cdot) \in \R^{\vert \mathcal{V} \vert}$ & Model output logits. \\
$y \in \mathcal{V}$ & An output token; $y \sim f(\cdot)$. \\
$\theta$ & Model parameters. \\
$m(\cdot)$ & Correct-class margin function. \\
$S_\rvx$ & Stability region of prompt $\rvx$. \\
$B_\epsilon(\rvh)$ & $\epsilon$-ball of embeddings $\rvh$; $B_\epsilon(\rvh) = \{ \rvh' \vert \| \rvh' - \rvh \|_2 < \epsilon \}$. \\
$d(\cdot, \cdot)$ & Distance between two prompts; edit distance by default. \\
$\mathrm{Init}(\cdot)$ & Prompt initialization for obfuscation optimization. \\
$\mathcal{L}(\cdot)$ & Objective function for obfuscation optimization. \\
\bottomrule
\end{tabular}
\end{table}

\section{Discussions on Theoretical Motivation for \ourshort{}} \label{app:theoretical_motivation}
In this section, we provide details that extend our discussions in \cref{subsec:theoretical_motivation}.
We will show the existence of obfuscated prompts that satisfy our criteria of obfuscation, utility and non-portability.

\mypar{Roadmap.}
We start with defining the notion of ``functional equivalence'' in the setting of single-token generation (\textit{utility of embeddings}).
We show that there always exists embeddings that is different from yet functionally equivalent to embeddings of the original prompt (\textit{local obfuscation and utility}).
We then show that an obfuscated prompt satisfying the utility constraint often fails to transfer to another model with parameters different from the target model (\textit{non-portability}).
Finally, we extend our theoretical results to the scenario of open-ended, autoregressive generation.

Let output logits of the target model be $f(\cdot;\theta)$.
Let original agent prompt embeddings be $\mathbf{h} = \mathrm{Embed}(\mathbf{x}) \in \R^{n \times d}$, i.e. embeddings with length $n \coloneqq \vert \rvx \vert$ and dimension $d$.
The obfuscated embeddings ($\tilde{\rvh}$) can be obtained by adding noise $\bm{\delta}$ to $\rvh$, such that $\tilde{\rvh} = \rvh+\bm{\delta}$.
We assume $\tilde{\rvx}$ has the same length as $\rvx$, which can always be achieved via padding either $\tilde{\rvx}$ or $\rvx$ at initialization.

For the moment, we focus on embeddings in continuous space, not prompts in discrete space.

By saying that embeddings $\tilde{\rvh}$ is functionally equivalent to $\rvh$, we refer to sampled model responses: for $\forall \rvq_i$, $\tilde{y}_i \sim f(\cdot \vert \tilde{\rvh},\rvq_i), y_i \sim f(\cdot \vert \rvh,\rvq_i)$ where $\tilde{y}_i, y_i \in \mathcal{V}$, we have $\tilde{y}_i = y_i$.
We adopt a relaxed definition of functional equivalence based on discrete output behavior and restrict sampling strategy to greedy decoding:
\begin{definition}[Functional equivalence of embeddings]\label{def:functional_equivalence_of_embeddings}
We say the embeddings $\tilde{\rvh}$ and $\rvh$ of two prompts ($\rvx$, $\tilde{\rvx}$) are functionally equivalent with respect to a query $\rvq_i$ if the same outputs are obtained:
\begin{equation}
\tilde{y}_i = \underset{y}{\argmax} f(y \vert \tilde{\mathbf{h}},\mathbf{q}_i),
y_i = \underset{y}{\argmax} f(y \vert \mathbf{h},\mathbf{q}_i).
\end{equation}
\end{definition}
This definition is less restrictive and more practical than representation-level equivalence, since in practice we do not care about hidden states, only final discrete outputs.

Following our definition of functional equivalence, we now establish the conditions under which such equivalent prompts exist and satisfy our objectives.

\mypar{Local obfuscation and utility.}
To prove that we can find an obfuscated prompt  that preserves utility, we first establish the geometric properties of the model's decision boundary in the embedding space.

\begin{lemma}[Equivalent condition for functional equivalence]
Let $y_i = \argmax_y f(y \vert \rvh, \rvq_i)$ be the target output for the original prompt. For an obfuscated embedding $\tilde{\rvh}$ to be functionally equivalent to $\rvh$, it is sufficient that the correct-class margin is strictly positive:

$$m(\tilde{\rvh}, \rvq_i, y_i) \coloneqq f(\tilde{\rvh}, \rvq_i)_{y_i} - \max_{k \neq y_i} f(\tilde{\rvh}, \rvq_i)_k > 0$$
\end{lemma}

\begin{proof}
If $m(\tilde{\rvh}, \rvq_i, y_i) > 0$, then $f(\tilde{\rvh}, \rvq_i)_{y_i} > f(\tilde{\rvh}, \rvq_i)_k$ for all $k \neq y_i$. Consequently, the argmax operation under greedy decoding yields $\tilde{y}_i = y_i$, satisfying \cref{def:functional_equivalence_of_embeddings}.
\end{proof}

We next show that this condition holds not just for a single point, but for a region surrounding .

\begin{lemma}[Continuity of the margin function]
The margin function $m(\rvh, \rvq_i, y_i)$ is continuous with respect to the input embeddings $\rvh$.
\end{lemma}

\begin{proof}
The neural network $f(\cdot; \theta)$ is a composition of continuous functions (linear transformations, layer normalizations, and activations like GeLU/Softmax). The operations $\max(\cdot)$ and subtraction are also continuous. Therefore, their composition $m(\cdot)$ is continuous with respect to $\rvh$.
\end{proof}

Using continuity, we define the region in embedding space where utility is preserved.

\begin{definition}[Stability region]
We define the stability region $S_{\rvx}$ as the set of all embeddings that maintain positive margins for the target outputs across a query set $\mathcal{Q}$:

$$S_{\rvx} = \{ \rvh' \mid \forall \rvq_i \in \mathcal{Q}, m(\rvh', \rvq_i, y_i) > 0 \}$$
\end{definition}

Since $m(\rvh, \dots) > 0$ for the original prompt (assuming the model is confident) and $m$ is continuous, there exists an $\epsilon > 0$ such that the open ball $B_{\epsilon}(\rvh) \subset S_{\rvx}$. Any embedding vector within this ball preserves utility.

\begin{theorem}[Existence of obfuscated prompts] \label{thm:existence_of_obfuscated_prompts}
Given a task prompt $\rvx$, there exists a prompt $\tilde{\rvx} \neq \rvx$ such that:

1. $\tilde{\rvx}$ is functionally equivalent to $\rvx$ (Utility).
2. The distance in prompt space $d(\tilde{\rvx}, \rvx) \geq d_0$ for some constant $d_0 > 0$ (Obfuscation).
\end{theorem}

\begin{proof}
Let $\Delta \rvh$ be the perturbation in embedding space caused by modifying $\rvx$ to $\tilde{\rvx}$. We require $\rvh + \Delta \rvh \in B_{\epsilon}(\rvh)$ to satisfy utility.

The total perturbation can be decomposed into token-wise perturbations. Let $\tilde{\rvx}$ differ from $\rvx$ at a set of indices $\mathcal{K}$. The norm of the embedding shift is bounded by the sum of individual token shifts: $\|\Delta \rvh\| \leq \sum_{j \in \mathcal{K}} \|\bm{\delta}_j\|$.

Due to the attention mechanism in Transformer architectures, specific tokens often exhibit low sensitivity (gradient $\nabla_{\rvx_j} \mathcal{L} \approx 0$), particularly in long contexts (attention dilution). We can select a set of $k$ such tokens to replace or perturb such that $\sum_{j=1}^k \|\bm{\delta}_j\| < \epsilon$.

While the embedding perturbation is bounded within $\epsilon$ (preserving utility), the discrete prompt distance (e.g., edit distance) $d(\tilde{\rvx}, \rvx)$ increases with $k$. By choosing sufficiently many low-sensitivity tokens, we satisfy $d(\tilde{\rvx}, \rvx) \geq d_0$ while ensuring $\tilde{\rvh} \in S_{\rvx}$.
\end{proof}

\mypar{Non-portability.}
Finally, we address why this utility does not transfer to other models.

\begin{proposition}[Non-portability via manifold mismatch]
An obfuscated prompt $\tilde{\rvx}$ that satisfies utility on model $\theta$ is likely to fail on a distinct model $\theta'$.
\end{proposition}

\begin{proof}
The stability region $S_{\rvx}(\theta)$ depends on the specific parameters $\theta$. For a different model $\theta'$, the decision boundaries and resulting stability region $S_{\rvx}(\theta')$ generally do not align with $S_{\rvx}(\theta)$ in the high-dimensional embedding space.

The perturbation $\bm{\delta}$ constructed in Theorem 1 is optimized specifically to remain within $B_{\epsilon}(\rvh) \subset S_{\rvx}(\theta)$. Without knowledge of $\theta'$, there is no guarantee that $\rvh + \bm{\delta} \in S_{\rvx}(\theta')$. Empirically, adversarial or obfuscated perturbations are known to be sensitive to the specific curvature of the loss landscape, leading to $P(\tilde{\rvh} \in S_{\rvx}(\theta')) \ll P(\tilde{\rvh} \in S_{\rvx}(\theta))$, ensuring non-portability.
\end{proof}

\mypar{Extension to autoregressive generation.}
We previously focus on generation of a single token; we will now generalize our results to the practical scenario of open-ended generation.

\begin{proposition}[Autoregressive extension] \label{prop:autoregressive_extension}
The guarantees of utility and obfuscation extend from single-token generation to open-ended autoregressive generation.
\end{proposition}

\begin{proof}
We proceed by induction on the generation step $t$.

\textbf{Base case ($t=0$):} Theorem 1 guarantees $\tilde{y}_0 = y_0$ given $\tilde{\rvx}$.

\textbf{Inductive step:} Assume $\tilde{y}_k = y_k$ for all $k < t$. At step $t$, the input context for the obfuscated prompt is $[\tilde{\rvx}, y_0, \dots, y_{t-1}]$. Since the generated history is identical, the deviation in embeddings stems solely from the initial $\tilde{\rvx}$. As established in Theorem 1, provided the initial perturbation is within the stability region, the subsequent token prediction remains invariant. Thus, $\tilde{y}_t = y_t$.
\end{proof}

\section{Details on Main Experiment} \label{app:main_experiment_details}
In this section, we provide implementation details for the experiments of \cref{subsec:main_results}.

\mypar{Prompt initialization: Code-symbol encoding.}
In \cref{fig:prompt_initialization_dsl_template}, we show the prompt template for the target LLM to encode the original prompt into a code-symbol format that is not written in natural language.

\begin{figure}
\centering
\begin{tcolorbox}[left=1.5mm, right=1.5mm, top=1.5mm, bottom=1.5mm, colback=black!5!white, colframe=black, title={Prompt template to encode original prompt into code-symbol format for prompt initialization.}]
\begin{lstlisting}[basicstyle=\scriptsize\ttfamily, columns=flexible, breaklines=true, escapeinside={(*@}{@*)}]
You are a large language model with strong pattern abstraction ability.
Your task is NOT to execute a prompt, but to DESIGN a prompt DSL that YOU can reliably interpret and execute.
Constraints:
1. The DSL must preserve the intent and usage effect of the original prompt template.
2. The DSL must be directly usable (no intermediate explanation required at runtime).
3. The DSL should rely on structural patterns.
4. The DSL should be difficult for humans to read directly.
5. Do not reduce too much semantic information.
6. Such as """<transmute operator="textflow" mode="bidirectional" direction="input(*@$\rightarrow$@*)en">
    <source lang="auto" detect="true">{{QUERY}}</source>
    <constraint preserve="semantic_density" threshold="0.92"/>
    <output format="natural" register="neutral"/>
</transmute>"""

Prompt:
"""{target_prompt}"""
\end{lstlisting}
\end{tcolorbox}
\caption{Prompt template for the target LLM to rewrite the original prompt into code-symbol format with XML-like structure.}
\label{fig:prompt_initialization_dsl_template}
\end{figure}

\mypar{Noise scheduling strategy.}
In the main body, we do not include the noise scheduling strategy for simplicity.
However, in practice, we find that \textit{noise annealing} is helpful for obfuscation optimization.
We interpret this strategy as an implementation detail rather than our core design.
Noise annealing is motivated by our empirical observation that it is hard to inject large amounts of noise at late stages of obfuscation optimization.
We hypothesize that obfuscation optimization should be conducted at finer granularities at later stages, whereas aggressive noise injection at earlier stages is helpful for improving efficiency.
Therefore we draw inspiration from \textit{learning rate annealing}, a prevailing practice in neural network optimization with stochastic gradient descent.

Specifically, we start with a preset number of random characters to inject sampled noise into, which we call the \textit{initial noise size} (analogous to initial learning rate).
After each epoch, we adopt a linear schedule and decrease the number of noise by a constant, which we term the \textit{noise schedule rate} (analogous to learning rate decay rate).
We also designate a \textit{minimum noise size} such that optimization does not stop completely at late stages of obfuscation optimization.

\mypar{Caveat on obfuscation optimization on API LLMs.}
In practice, obfuscation optimization is hindered by the lack of acc to full logprobs for API LLMs such as GPT-4o and DeepSeek-chat.
For example, DeepSeek API\footnote{\url{https://api-docs.deepseek.com/api/create-chat-completion}} only returns the top logprobs of at most 20 tokens (\verb|top_logprobs| argument).
This slightly complicates the computation of the task loss.

The algorithm for computing cross-entropy task loss is shown in \cref{alg:task_loss_api_model}.
We set a default logprob (often -100.0) for label tokens missing from the returned top-$k$ logprobs.
We take advantage of the top-logprobs functionality of API LLMs while balancing efficiency and performance.
As is shown in \cref{tab:hparams_main_results}, we let $k=10$ since a large $k$ can greatly slow down API response rate while a small $k$ decreases the precision of loss computation.

\begin{algorithm}
\caption{Computation of task loss for API LLMs.}
\label{alg:task_loss_api_model}
\begin{algorithmic}
\STATE {\bfseries Input:} Base LLM $p(\cdot)$, obfuscated prompt $\tilde{\rvx}_t$ at the $t$-th step, query $\rvq_t$ at the $t$-th step, labels $\rvy_t$ at the $t$-th step, number of top logprobs $k$, default logprob $l_\text{def}$
\STATE {\bfseries Output:} Task loss $l_t$
\STATE $\mathrm{logps} \gets p(\cdot \vert \rvq_t, \tilde{\rvx}_t)$
\STATE $l_t \gets 0$
\STATE $i = 1$
\WHILE{$i \leq \vert \rvy_t \vert$}
    \IF{$\rvy_t[i]$ in top-$k$ logprobs}
        \STATE $l_t \gets l_t -\mathrm{logps}[i][\rvy_t[i]]$
    \ELSE
        \STATE $l_t \gets l_t - l_\text{def}$
    \ENDIF
    \STATE $i \gets i + 1$
\ENDWHILE
\STATE $l_t \gets l_t / \vert \rvy_t \vert$
\end{algorithmic}
\end{algorithm}

\mypar{Hyperparameters.}
We show the hyperparameters of our main experiment in \cref{tab:hparams_main_results}.

\begin{table}
\caption{Hyperparameters for \ourshort{} in main experiment.}
\label{tab:hparams_main_results}
\centering
\footnotesize
\begin{tabular}{lc}
\toprule
Hyperparameter & Value \\
\midrule
$\lambda$ & 0.1 \\
$\gamma$ & 0.1 \\
Epochs & 50 \\
Candidates per epoch & 20 \\
Noise set & ASCII characters (0--127) \\
Initial noise size & \nicefrac{1}{4} prompt length \\
Noise schedule rate & 8 \\
Minimum noise size & 4 \\
Top-$k$ logprobs & 10 \\
Default logprob & -100.0 \\
\bottomrule
\end{tabular}
\end{table}

\section{Details on Case Study}
\label{app:details-on-case-study}

\mypar{Reproducibility.}
In \cref{subsec:case_study} of the main body, we present an obfuscated prompt and its corresponding original prompt as well as model responses from the target model and non-target model.
In order to facilitate reproduction and to support the authenticity of our results, we have provided \textbf{an executable demonstration of the case study in the supplementary material}.
We encourage readers to execute the Python script for themselves, which requires access to the DeepSeek API\footnote{\url{https://api-docs.deepseek.com/}}.

\section{Details on Adaptive Attacks} \label{app:adaptive_attack_details}
In this section, we provide implementation details for the adaptive attack experiments of \cref{subsec:adaptive_attack}.

\subsection{LLM-Assisted Prompt Recovery} \label{subapp:llm_assisted_recovery_details}

We show the prompt template for LLM-assisted prompt recovery in \cref{fig:llm_assisted_recovery_prompt_template}.
The LLM is provided with the obfuscated prompt and instructed to guess the original prompt.
Although we acknowledge that this implementation is simple, we emphasize that this strategy directly supports our claim that even the target LLM for which the obfuscated prompt was optimized fails to fully interpret the obfuscated prompt.

\begin{figure}
\centering
\begin{tcolorbox}[left=1.5mm, right=1.5mm, top=1.5mm, bottom=1.5mm, colback=black!5!white, colframe=black, title={Prompt template for LLM-assisted prompt recovery.}]
\begin{lstlisting}[basicstyle=\scriptsize\ttfamily, columns=flexible, breaklines=true]
What is the original prompt for the following prompt?

Prompt:
"""{obfuscated_prompt}"""
\end{lstlisting}
\end{tcolorbox}
\caption{Prompt template for LLM-assisted prompt recovery.}
\label{fig:llm_assisted_recovery_prompt_template}
\end{figure}

\begin{figure}
\centering
\begin{tcolorbox}[
  left=1.5mm, right=1.5mm, top=1.5mm, bottom=1.5mm,
  colback=black!5!white, colframe=black,
  title={Prompt template for naive prompt baseline (LLM-assisted prompt induction).}
]
\begin{lstlisting}[basicstyle=\scriptsize\ttfamily, columns=flexible, breaklines=true]
You are given a small set of input-output examples produced by an agent.
Your task is to write a SINGLE system prompt that, when used to initialize an LLM agent,
will reproduce the same behavior on similar inputs.

[Optional] Task context:
{task_description}

Examples:
{io_pairs}
# Each example should follow:
# Input:  ...
# Output: ...

Requirements for the system prompt:
- Output ONLY the system prompt text (no explanations, no markdown).
- The prompt should specify the agent's role, goals, constraints, and tool-use policy if needed.
- Be concise but complete; preserve the behavior implied by the examples.
- Do NOT refer to these examples explicitly in the final prompt.

System prompt:
\end{lstlisting}
\end{tcolorbox}
\caption{Template used to construct the naive prompt baseline via LLM-assisted prompt induction from observed input-output behavior.}
\label{fig:naive_prompt_template}
\end{figure}

\subsection{Deobfuscation Attack} \label{subapp:deobfuscation_details}
\paragraph{Deobfuscation optimization.} We show the hyperparameters for deobfuscation attack in \cref{tab:hparams_deobfuscation}.
All hyperparameters are consistent with those of obfuscation optimization, except for the coefficient of the optimization objective.
The objective function for the deobfuscation attack includes only $\mathcal{L}_\text{task}$ and $\mathcal{L}_\text{non-lang}$ terms, not $\mathcal{L}_\text{dist}$, since the original prompt is inaccessible for the attacker.
We also invert the coefficient for the $\mathcal{L}_\text{non-lang}$ term since the attacker aims to recover the original prompt written in natural language.

\paragraph{Naive Prompt Baseline.}
We also consider a naive baseline where the attacker re-synthesizes a system prompt from task behavior, rather than deobfuscating the deployed prompt. Specifically, the attacker collects a small set of representative input--output pairs via black-box queries and uses an LLM to draft a prompt that reproduces the observed functionality. The drafted prompt is evaluated directly without any adaptive tuning or optimization. The concrete induction template used to generate the naive prompt is provided in \cref{fig:naive_prompt_template}.

\begin{table}
\caption{Hyperparameters for deobfuscation attack.}
\label{tab:hparams_deobfuscation}
\centering
\footnotesize
\begin{tabular}{lc}
\toprule
Hyperparameter & Value \\
\midrule
$\gamma$ & -0.1 \\
Epochs & 50 \\
Candidates per epoch & 20 \\
Noise set & ASCII characters (0--127) \\
Initial noise size & \nicefrac{1}{4} prompt length \\
Noise schedule rate & 8 \\
Minimum noise size & 4 \\
Top-$k$ logprobs & 10 \\
\bottomrule
\end{tabular}
\end{table}

\section{Discussions on \ourshort{} Methodology}
\label{app:discussions-on-methodology}
\mypar{On \ourshort{} vs. prompt compression techniques.}
Historically, prompt compression is motivated by maintaining task performance comparable to original prompts with less memory and computational cost.
\ourshort{} might be reminiscent of \textit{hard prompt compression methods}, according to the taxonomy of \citet{li2025prompt}.
However, these methods focus on removing redundancy from the original prompts and often require auxiliary compression models.
Furthermore, hard prompt compression methods do not take obfuscation into consideration.
All these properties clearly distinguish \ourshort{} from prompt compression techniques.

\mypar{Concerns of invertibility.}
A recent paper claims that transformer language models are almost surely injective and thus invertible (i.e. to recover the input token sequence given hidden state activations)~\citep{nikolaou2025language}, which seems to challenge our prompt protection method.
However, we argue that the arguments of this paper do not fundamentally undermine the validity of \ourshort{}; instead, \textbf{it strengthens the security of \ourshort{}}.

Given a task query $\rvq$ and a prompt $\rvx$, there are countless hidden states that yield the same greedy output as $\rvh = f_{<l}(\rvx)$ where $f_{<l}(\cdot)$ is the first $l$ layers of the LM, since the latent space is continuous and the LM has inherent robustness with respect to infinitesimal perturbations to hidden states.
All these functionally equivalent hidden states belong to a set of activations: $B_\epsilon = \{\rvh' \vert \| \rvh' - \rvh \|_2 < \epsilon \}$, where $\epsilon>0$ is a small constant that ensures greedy model outputs are the same.
Suppose it is possible to trace back each hidden state of $B_\epsilon$ back to input token sequences, then there are infinite number of prompts (assuming model context size is intractably large and we use paddings to handle various prompt lengths) that lead to the same model outputs, since $\vert B_\epsilon \vert$ is infinite.
This makes the prompt inversion attack intractable, since it is impossible for the attacker to identify the prompt $\rvx$ among the almost infinite set of functionally equivalent prompts.


\end{document}